\RequirePackage{fix-cm}
\documentclass{article}                     % onecolumn (standard format)
\usepackage{graphicx}
\usepackage{mathptmx} 
\usepackage[ruled,linesnumbered]{algorithm2e}
\usepackage{tikz}
\usetikzlibrary{shapes, graphs, fit}
\usepackage{amsmath}
\usepackage{amsfonts}
\usepackage{amssymb}
\usepackage{amsthm}
\usepackage[sort,comma]{natbib}
\usepackage{hyperref}
\theoremstyle{definition}
\newtheorem{example}{Example}[section]
\newtheorem{theorem}{Theorem}[section]
\newtheorem{corollary}{Corollary}[section]
\newtheorem{lemma}{Lemma}[section]

\begin{document}

\title{Stochastic Enumeration with Importance Sampling}


\author{Alathea Jensen}


%\institute{Alathea Jensen \at
%              George Mason University \\
%              Fairfax, VA
%}

\date{December 5, 2017}

\maketitle

\begin{abstract}
Many hard problems in the computational sciences are equivalent to counting the leaves of a decision tree, or, more generally, by summing a cost function over the nodes.  These problems include calculating the permanent of a matrix, finding the volume of a convex polyhedron, and counting the number of linear extensions of a partially ordered set.  Many approximation algorithms exist to estimate such sums. One of the most recent is Stochastic Enumeration (SE), introduced in 2013 by Rubinstein. In 2015, Vaisman and Kroese provided a rigorous analysis of the variance of SE, and showed that SE can be extended to a fully polynomial randomized approximation scheme for certain cost functions on random trees.  We present an algorithm that incorporates an importance function into SE, and provide theoretical analysis of its efficacy. We also present the results of numerical experiments to measure the variance of an application of the algorithm to the problem of counting linear extensions of a poset, and show that introducing importance sampling results in a significant reduction of variance as compared to the original version of SE.
%\keywords{Randomized algorithms \and Monte Carlo sampling \and Importance sampling \and Sequential Importance Sampling \and Linear extensions \and Decision tree \and Counting}

%\subclass{05C05 \and 65C05 \and 05C85 \and 05C81 \and 60J80}
\end{abstract}

\section*{Acknowledgments}

This is a pre-print of an article published in Methodology and Computing in Applied Probability. The final authenticated version is available online at:

\url{https://doi.org/10.1007/s11009-018-9619-2}

The author would like to thank Isabel Beichl and Francis Sullivan for the idea for this project.  The author would also like to thank the Applied and Computational Mathematics Division of the Information Technology Laboratory at the National Institute of Standards and Technology for hosting the author as a guest researcher during the preparation of this article.

\section{Introduction}
Many hard problems in mathematics, computer science, and the physical sciences are equivalent to summing a cost function over a tree.  These problems include calculating the permanent of a matrix, finding the volume of a convex polyhedron, and counting the number of linear extensions of a partially ordered set.

There are tree-searching algorithms which give an exact answer by simply traversing every node in the tree; however, in many cases, the tree is far too large for this to be practical.  Indeed, the problem of computing tree cost is in the complexity class \#P-complete \citep*{valiant}.  This complexity class consists of counting problems which find the number of solutions that satisfy a corresponding NP-complete decision problem.

Accordingly, there are various approximation algorithms for tree cost, and the two main types of these are Markov Chain Monte Carlo (MCMC) and sequential importance sampling (SIS).  Both of these perform random sampling on a suitably defined set.

The original version of SIS is Knuth's algorithm \citep*{knuth}, which samples tree cost by walking a random path from the root to a leaf, where each node in the path is chosen uniformly from the children of the previously chosen node.  There have been several major adaptations to Knuth's algorithm, all of which attempt to reduce the variance of the estimates produced.

One modification of Knuth's algorithm is to choose the nodes of the path non-uniformly, proportional to an importance function on the nodes.  Of course, choosing a good importance function requires some knowledge about the structure of the tree, and so this approach is not suitable for random trees, but rather for families of trees which share some general characteristics.  Some cases where this approach has produced good results can be found in \citet*{jcp}, \citet*{blitz}, \citet*{algo}, \citet*{karp}, for example.

There have also been adaptations of Knuth's algorithm which change the algorithm in a more structural way, such as stratified sampling, which was introduced by Knuth's student, \citet*{chen}.

Stochastic Enumeration (SE) is the most recent of the structural adaptations.  It was originally introduced by \citet*{rubin}, and further developed in \citet*{rubinbook}.  Its approach to the problem is to run many non-independent trajectories through the tree in parallel, combining their effect on the estimate at each level of the tree to produce a single final estimate of the tree cost.  Alternatively, one can view SE as operating on a hypertree associated with the original tree.  A similar approach to the problem was taken by \citet*{brian}.

In Rubinstein's original definition, the SE algorithm was only able to count the leaves of a tree.  \citet*{vaisman} updated SE to estimate tree cost for any cost function, and provided a rigorous analysis of the variance.  They also showed that SE can be extended to an fully polynomial randomized approximation scheme (FPRAS) for random trees with a cost function that is 1 on every node.

In this paper, we follow up on the work of Vaisman and Kroese to develop an adaptation of SE which we call Stochastic Enumeration with Importance (SEI).  This algorithm chooses paths through the tree with non-uniform probability, according to a user-defined importance function on the nodes of the tree.  We provide a detailed analysis of the theoretical properties of the algorithm, including ways to bound the variance.

Just as with SIS, SEI is not suitable for random trees, but rather for families of trees which share some characteristics.  Therefore, in addition to theoretical analysis in which the importance function is not specified, we also provide a detailed example, with numerical results, of a family of trees and importance functions for which SEI provides a lower variance than SE.

\section{Definitions and Preliminaries}
Consider a tree $T$ with node set $\mathcal{V}$, where each node $v$ has some \textit{cost} $c(v)$ given by a \textit{cost function} $c:\mathcal{V}\to \mathbb{R}_{\geq 0}$.  We wish to estimate the total \textit{cost of the tree}, denoted $\mathrm{Cost}(T)$ and given by
\[\mathrm{Cost}(T)=\sum_{v\in\mathcal{V}}c(v)\]

If our tree is uniform, in the sense that all the nodes on a given level have the same number of children, then it is very easy to determine the number of nodes on each level.

We will call the root node level 0, the root's children level 1, and so on.  Suppose the root has $D_0$ children, the root's children all have $D_1$ children, and so on.  Then there is 1 node on level 0, $D_0$ nodes on level 1, $D_0 D_1$ nodes on level 2, and, in general, $D_0 D_1\cdots D_{i-1}$ nodes on level $i$.

If the cost of nodes is also uniform across each level, then we can easily add up the cost of the entire tree.  For each level $i$, let the cost of any node on level $i$ be denoted $c_i$.  Then the cost of our tree is

\begin{equation}\label{knuthcost}
\mathrm{Cost}(T)=c_0+c_1 D_0+c_2 D_0 D_1 + \dots + c_n D_0 D_1 \cdots D_{n-1}
\end{equation}
where $n$ is the lowest level of the tree.

Of course, most trees are not uniform is the sense described above, but the central idea of Knuth's algorithm \citep*{knuth} for estimating tree cost is to pretend as though they are.  In Knuth's algorithm, we walk a single path from the root to a leaf, and note the number of children that we see from each node in our path ($D_0, D_1, \dots, D_n$), as well as the cost of each node in our path ($c_0, c_1, \dots, c_n$).  We then calculate the cost of the tree using Formula (\ref{knuthcost}), which is no longer exact but is now an unbiased estimator of the tree cost.

In the SE algorithm, just as in Knuth's algorithm, we work our way down the tree level by level from the root to the leaves.  The main difference is that instead of choosing a single node on each level of the tree, we choose multiple nodes on each level.  We can also think of this as choosing a single hypernode from each level of a hypertree constructed from the original tree.  The following definitions are necessary to describe the structure of the hypertree.

We define a \textit{hypernode} to be a set of distinct nodes $\mathbf{v}=\{v_1,\dots,v_m\}\subset\mathcal{V}$ that are in the same level of the tree.  We can extend the definition of the cost function to hypernodes by letting
\[c(\mathbf{v})=\sum_{v\in\mathbf{v}}c(v)\]

Let $S(v)$ denote the set of successors (or children) of a node in the tree.  Then we can define the set of \textit{successors of a hypernode} $\mathbf{v}$ as
\[S(\mathbf{v})=\bigcup_{v\in\mathbf{v}}S(v)\]

Throughout the SE algorithm, each time we move to a new level, we choose a new hypernode from among the successors $S(\mathbf{v})$ of the previous hypernode $\mathbf{v}$.  We make no distinction between these successors in terms of which node in the previous hypernode they came from.  This means that some nodes in the previous hypernode may have multiple children chosen to be in the next hypernode, while other nodes in the previous hypernode may not have any children chosen to be in the next hypernode.

Obviously there is some limit on our computing power, so we have to limit the size of the hypernodes we work with to be within a budget, which we will denote $B\in\mathbb{N}$.  At each level, we will choose $B$ nodes to be in the next hypernode, as long as $S(\mathbf{v})$ is larger than $B$.  If $|S(\mathbf{v})|\leq B$, then we will take all of $S(\mathbf{v})$ to be the next hypernode.

Thus, if our current hypernode is $\mathbf{v}$, the candidates for our next hypernode, which we call the \textit{hyperchildren} of $\mathbf{v}$, are the elements of the set
\[H(\mathbf{v})=\big\{\mathbf{w}\subseteq S(\mathbf{v}):|\mathbf{w}|=\mathrm{min}(B,|S(\mathbf{v})|)\big\}\]

Many of the statements and proofs throughout this paper are in a recursive form that refers to subforests of a tree, and so we lastly need to define a \textit{forest rooted at a hypernode}.  For a hypernode $\mathbf{v}$, the forest rooted at $\mathbf{v}$, denoted $T_{\mathbf{v}}$, is simply the union of all the trees rooted at each of the nodes in $\mathbf{v}$.
\[T_\mathbf{v}=\bigcup_{v\in\mathbf{v}}T_v\]

We can also extend the notion of the total cost of a tree to a forest rooted at a hypernode by letting
\[\mathrm{Cost}(T_\mathbf{v})=\sum_{v\in\mathbf{v}}\mathrm{Cost}(T_v)\]

Let's look at an example to familiarize ourselves further with the notation.

\begin{example}\label{notationexample}
Consider the tree in Figure \ref{notationtree}.  It is labeled with a possible sequence of hypernodes that could be chosen by the SE algorithm, using a budget of $B=2$.

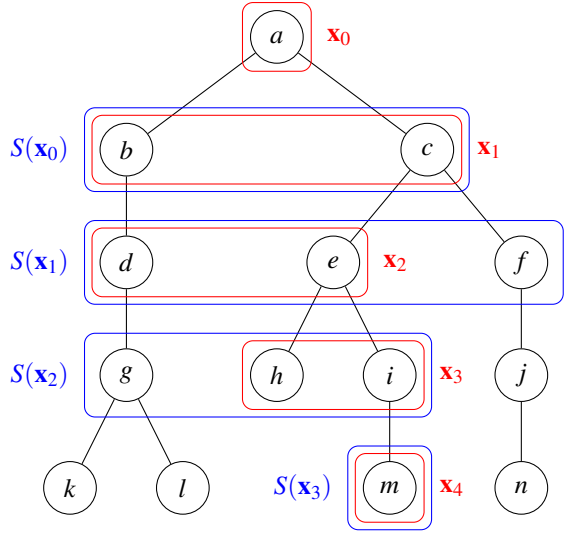
\begin{figure}
\centering
\begin{tikzpicture}[
level distance = 15mm,
level 1/.style={sibling distance=40mm},
level 2/.style={sibling distance=25mm},
level 3/.style={sibling distance=15mm},
every node/.style={draw,circle, minimum size=7mm}]
\node(a) {$a$}
child {
	node(b) {$b$}
	child {
		node(d) {$d$}
		child {
			node(g) {$g$}
			child {
				node(k) {$k$}
				}
			child {node(l) {$l$}}
			}
		}
	}
child {
	node(c){$c$}
	child {
		node(e) {$e$}
		child {	node(h) {$h$}}		
		child {
			node(i) {$i$}
			child {node(m) {$m$}}
			}
		}
	child {
		node(f) {$f$}
		child {
			node(j) {$j$}
			child { node(n) {$n$} }
			}
		}
	};
\node[draw=red, fit=(a),inner sep=1mm,rectangle,rounded corners,label={[red]right:$\mathbf{x}_0$}](x0){};
\node[draw=red, fit=(b) (c),inner sep=1mm,rectangle,rounded corners,label={[red]right:$\mathbf{x}_1$}](x1){};
\node[draw=red, fit=(d) (e),inner sep=1mm,rectangle,rounded corners,label={[red]right:$\mathbf{x}_2$}](x2){};
\node[draw=red, fit=(h) (i),inner sep=1mm,rectangle,rounded corners,label={[red]right:$\mathbf{x}_3$}](x3){};
\node[draw=red, fit=(m),inner sep=1mm,rectangle,rounded corners,label={[red]right:$\mathbf{x}_4$}](x4){};

\node[draw=blue, fit=(b) (c),inner sep=2mm,rectangle,rounded corners,label={[blue]left:$S(\mathbf{x}_0)$}](Sx0){};
\node[draw=blue, fit=(d) (e) (f),inner sep=2mm,rectangle,rounded corners,label={[blue]left:$S(\mathbf{x}_1)$}](Sx1){};
\node[draw=blue, fit=(g) (h) (i),inner sep=2mm,rectangle,rounded corners,label={[blue]left:$S(\mathbf{x}_2)$}](Sx2){};
\node[draw=blue, fit=(m),inner sep=2mm,rectangle,rounded corners,label={[blue]left:$S(\mathbf{x}_3)$}](Sx3){};
\end{tikzpicture}
\caption{Tree for Example \ref{notationexample}, with each chosen hypernode boxed and labeled to the right, and each chosen hypernode's successor set boxed and labeled to the left}
\label{notationtree}
\end{figure}

On level 0, the root is automatically chosen to be the first hypernode, $\mathbf{x}_0$.  We could then refer to the entire tree as $T_{\mathbf{x}_0}$.  On level 1, we have $S(\mathbf{x}_0)=\{b,c\}$.  Since $|S(\mathbf{x}_0)|\leq B$, we take all of $S(\mathbf{x}_0)$ to be our next hypernode, so $\mathbf{x}_1=\{b,c\}$.

On level 2, we have $S(\mathbf{x}_1)=\{d,e,f\}$, so our choices for $\mathbf{x}_2$ are the elements of $H(\mathbf{x}_1)=\{\{d,e\},\{d,f\},\{e,f\}\}$.  Let's choose $\mathbf{x}_2=\{d,e\}$.  Similarly, on level 3, we have $S(\mathbf{x}_2)=\{g,h,i\}$, so our choices for $\mathbf{x}_3$ are $H(\mathbf{x}_2)=\{\{g,h\},\{g,i\},\{h,i\}\}$.  Let's choose $\mathbf{x}_3=\{h,i\}$.

Finally, on level 4, we have $S(\mathbf{x}_3)=\{m\}$.  Since $|S(\mathbf{x}_3)|\leq B$, we take all of $S(\mathbf{x}_3)$ to be our next hypernode, so $\mathbf{x}_4=\{m\}$.\qed
\end{example}

\section{Stochastic Enumeration with Arbitrary Probability}

We are now ready to state the first algorithm, Stochastic Enumeration with arbitrary probability (SEP).  It is a generalization of the updated Stochastic Enumeration algorithm in \citet*{vaisman}, which used uniform probabilities.

\begin{algorithm}[H]
\DontPrintSemicolon
\SetKwInOut{Input}{Input}\SetKwInOut{Output}{Output}
\Input{A forest $T_\mathbf{v}$ of height $h$ rooted at a hypernode $\mathbf{v}$, and a budget $B\in\mathbb{N}$}
\Output{An unbiased estimator $|\mathbf{v}|C_{\mathrm{SEP}}$ of the total cost of the forest $T_\mathbf{v}$}
\BlankLine
\textbf{(Initialization):} Set $k\leftarrow 0$, $D\leftarrow 1$, $\mathbf{x}_0=\mathbf{v}$, and $C_{\mathrm{SEP}}\leftarrow c(\mathbf{x}_0)/|\mathbf{x}_0|$.\;
\textbf{(Compute the successors):} Let $S(\mathbf{x}_k)$ be the set of all successors of $\mathbf{x}_k$.\;
\textbf{(Terminal position?):} If $|S(\mathbf{x}_k)|=0$, the algorithm stops, returning $|\mathbf{v}|C_{\mathrm{SEP}}$ as an estimator of $\mathrm{Cost}(T_\mathbf{v})$.\;
\textbf{(Advance):} Choose hypernode $\mathbf{x}_{k+1}\in H(\mathbf{x}_k)$ with probability $P(\mathbf{x}_{k+1})$.\;
\textbf{(Update):} Set $D_k\leftarrow \frac{|\mathbf{x}_{k+1}|}{|\mathbf{x}_k|}\binom{|S(\mathbf{x}_k)|-1}{|\mathbf{x}_{k+1}|-1}^{-1}(P(\mathbf{x}_{k+1}))^{-1}$, set $D\leftarrow D\cdot D_k$, and set $C_{\mathrm{SEP}}\leftarrow C_{\mathrm{SEP}}+\frac{c(\mathbf{x}_{k+1})}{|\mathbf{x}_{k+1}|}D$.\;
\textbf{(Loop):} Increase $k$ by 1 and return to Step 2.\;
\caption{Stochastic Enumeration with arbitrary probability (SEP) algorithm for estimating the cost of a backtrack tree}
\end{algorithm}
\vspace{\baselineskip}

Note that the quantity $D_k$ is an estimate of the number of children of the nodes in level $k$, so that after each update in line 5, $D$ is an estimate of the number of nodes in level $k+1$ of the tree.

Likewise, the quantity $\frac{c(\mathbf{x}_{k+1})}{|\mathbf{x}_{k+1}|}$ is an estimate of the average cost of nodes on level $k+1$, so that by adding $\frac{c(\mathbf{x}_{k+1})}{|\mathbf{x}_{k+1}|}D$ to $C_{\mathrm{SEP}}$ on line 5, we are adding the estimated cost of all of level $k+1$ of the tree.

Before analyzing this algorithm further, let's look at an example to get a better idea of how it works.

\begin{example}
\label{SEPexample}
Consider the tree in Figure \ref{SEPexamplefigure}.  To keep things simple, we'll use a budget of $B=2$ and a cost function $c$ that is 1 on every node.  Clearly the total cost of the tree is the number of nodes, 14.  This choice simplifies $\frac{c(\mathbf{x}_{k+1})}{|\mathbf{x}_{k+1}|}$ to 1, so the update command for $C_\mathrm{SEP}$ becomes 
\[C_{\mathrm{SEP}}\leftarrow C_{\mathrm{SEP}}+D\]

Let's choose hypernodes with a uniform probability, meaning $P(\mathbf{x}_{k+1})=1/|H(\mathbf{x}_k)|$.  Since $|H(\mathbf{x}_k)|=\binom{|S(\mathbf{x}_k)|}{|\mathbf{x}_{k+1}|}$, this makes the formula for $D_k$ simplify to $\frac{|S(\mathbf{x}_{k})|}{|\mathbf{x}_k|}$, so the update command for $D$ becomes
\[D\leftarrow\frac{|S(\mathbf{x}_{k})|}{|\mathbf{x}_k|}D\]

Note that $\frac{|S(\mathbf{x}_{k})|}{|\mathbf{x}_k|}$ is the average number of children of the nodes in $\mathbf{x}_k$.  In the original SE algorithm, the update command for $D$ always looks like this.

Now let's examine a possible sequence of hypernodes produced by Algorithm 2, as shown in Figure \ref{SEPexamplefigure}, which is the same as the previous example.

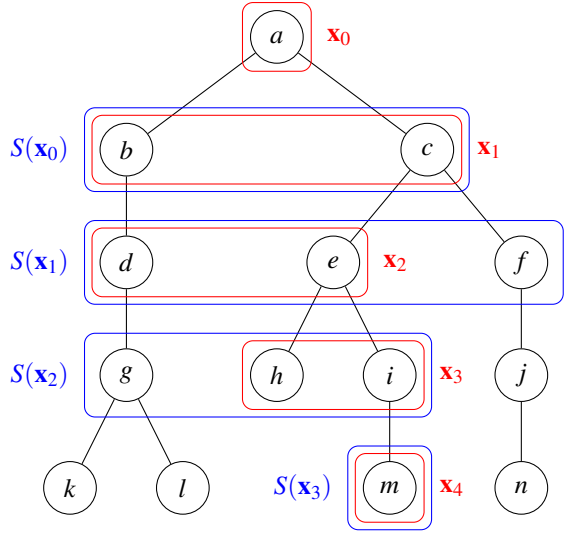
\begin{figure}
\centering
\begin{tikzpicture}[
level distance = 15mm,
level 1/.style={sibling distance=40mm},
level 2/.style={sibling distance=25mm},
level 3/.style={sibling distance=15mm},
every node/.style={draw,circle, minimum size=7mm}]
\node(a) {$a$}
child {
	node(b) {$b$}
	child {
		node(d) {$d$}
		child {
			node(g) {$g$}
			child {
				node(k) {$k$}
				}
			child {node(l) {$l$}}
			}
		}
	}
child {
	node(c){$c$}
	child {
		node(e) {$e$}
		child {	node(h) {$h$}}		
		child {
			node(i) {$i$}
			child {node(m) {$m$}}
			}
		}
	child {
		node(f) {$f$}
		child {
			node(j) {$j$}
			child { node(n) {$n$} }
			}
		}
	};
\node[draw=red, fit=(a),inner sep=1mm,rectangle,rounded corners,label={[red]right:$\mathbf{x}_0$}](x0){};
\node[draw=red, fit=(b) (c),inner sep=1mm,rectangle,rounded corners,label={[red]right:$\mathbf{x}_1$}](x1){};
\node[draw=red, fit=(d) (e),inner sep=1mm,rectangle,rounded corners,label={[red]right:$\mathbf{x}_2$}](x2){};
\node[draw=red, fit=(h) (i),inner sep=1mm,rectangle,rounded corners,label={[red]right:$\mathbf{x}_3$}](x3){};
\node[draw=red, fit=(m),inner sep=1mm,rectangle,rounded corners,label={[red]right:$\mathbf{x}_4$}](x4){};

\node[draw=blue, fit=(b) (c),inner sep=2mm,rectangle,rounded corners,label={[blue]left:$S(\mathbf{x}_0)$}](Sx0){};
\node[draw=blue, fit=(d) (e) (f),inner sep=2mm,rectangle,rounded corners,label={[blue]left:$S(\mathbf{x}_1)$}](Sx1){};
\node[draw=blue, fit=(g) (h) (i),inner sep=2mm,rectangle,rounded corners,label={[blue]left:$S(\mathbf{x}_2)$}](Sx2){};
\node[draw=blue, fit=(m),inner sep=2mm,rectangle,rounded corners,label={[blue]left:$S(\mathbf{x}_3)$}](Sx3){};
\end{tikzpicture}
\caption{Tree for Example \ref{SEPexample}, with each chosen hypernode boxed and labeled to the right, and each chosen hypernode's successor set boxed and labeled to the left}
\label{SEPexamplefigure}
\end{figure}

We initialize with $k=0$, $\mathbf{x}_0=\{a\}$, $D=1$, $C_{\mathrm{SEP}}=1$.  Then we compute $S(\mathbf{x}_0)=\{b,c\}$, which means $H(\mathbf{x}_0)=\{\{b,c\}\}$, and advance to $\mathbf{x}_1=\{b,c\}$ with $P(\mathbf{x}_1)=1$.  We update
\[D\leftarrow\frac{|S(\mathbf{x}_{0})|}{|\mathbf{x}_0|}D=2\]
\[C_{\mathrm{SEP}}\leftarrow C_{\mathrm{SEP}}+D=3\]

We advance to $k=1$ and loop.  We compute $S(\mathbf{x}_1)=\{d,e,f\}$, which means $H(\mathbf{x}_1)=\{\{d,e\},\{d,f\},\{e,f\}\}$, and advance to $\mathbf{x}_2=\{d,e\}$ with $P(\mathbf{x}_2)=\frac{1}{3}$.  We update
\[D\leftarrow\frac{|S(\mathbf{x}_{1})|}{|\mathbf{x}_1|}D=3\]
\[C_{\mathrm{SEP}}\leftarrow C_{\mathrm{SEP}}+D=6\]

We advance to $k=2$ and loop.  We compute $S(\mathbf{x}_2)=\{g,h,i\}$, which means $H(\mathbf{x}_2)=\{\{g,h\},\{g,i\},\{h,i\}\}$, and we advance to $\mathbf{x}_3=\{h,i\}$ with $P(\mathbf{x}_3)=\frac{1}{3}$.  We update
\[D\leftarrow\frac{|S(\mathbf{x}_{2})|}{|\mathbf{x}_2|}D=4.5\]
\[C_{\mathrm{SEP}}\leftarrow C_{\mathrm{SEP}}+D=10.5\]

We advance to $k=3$ and loop.  We compute $S(\mathbf{x}_3)=\{m\}$, which means $H(\mathbf{x}_3)=\{\{m\}\}$, and we advance to $\mathbf{x}_4=\{m\}$ with $P(\mathbf{x}_1)=1$.  We update
\[D\leftarrow\frac{|S(\mathbf{x}_{3})|}{|\mathbf{x}_3|}D=2.25\]
\[C_{\mathrm{SEP}}\leftarrow C_{\mathrm{SEP}}+D=12.75\]

We increase to $k=4$ and loop.  We compute $S(\mathbf{x}_4)=\emptyset$, so we are in the terminal position and we stop.  The algorithm returns $|\mathbf{x}_0|C_{\mathrm{SEP}}=12.75$ as an estimator of the cost of the tree.  This completes the example.\qed
\end{example}

Now we begin our analysis of Algorithm 1.  In general, the output of Algorithm 1 is a random variable
\begin{equation*}
\begin{split}
C_{\mathrm{SEP}}(T_{\mathbf{x}_0}) &= \frac{c(\mathbf{x}_0)}{|\mathbf{x}_0|}+{D_0}\frac{c(\mathbf{x}_1)}{|\mathbf{x}_1|}+{D_0 D_1}\frac{c(\mathbf{x}_2)}{|\mathbf{x}_2|}+\cdots+{D_0 D_1 \cdots D_{\tau-1}}\frac{c(\mathbf{x}_\tau)}{|\mathbf{x}_\tau|}\\
&= \frac{c(\mathbf{x}_0)}{|\mathbf{x}_0|}+{D_0}\left(\frac{c(\mathbf{x}_1)}{|\mathbf{x}_1|}+{D_1}\frac{c(\mathbf{x}_2)}{|\mathbf{x}_2|}+\cdots+{D_2 \cdots D_{\tau-1}}\frac{c(\mathbf{x}_\tau)}{|\mathbf{x}_\tau|}\right)\\
\end{split}
\end{equation*}
where $\tau$ is some height less than or equal to the height of $T_{\mathbf{x}_0}$.

This naturally suggests a recursive formulation of the output,
\[C_{\mathrm{SEP}}(T_{\mathbf{x}_0})= \frac{c(\mathbf{x}_0)}{|\mathbf{x}_0|}+{D_0}C_{\mathrm{SEP}}(T_{\mathbf{x}_1})\]

Let $\mathbf{w}$ be a hyperchild of $\mathbf{v}$ selected from $H(\mathbf{v})$ with probability $P(\mathbf{w})$.  Then we have
\begin{equation}\label{recursiveoutputeqn}
\begin{split}
C_{\mathrm{SEP}}(T_{\mathbf{v}})
&=\frac{c(\mathbf{v})}{|\mathbf{v}|}+{D_0}C_{\mathrm{SEP}}(T_{\mathbf{w}})\\
&=\frac{c(\mathbf{v})}{|\mathbf{v}|}+\frac{|\mathbf{w}|C_{\mathrm{SEP}}(T_{\mathbf{w}})}{|\mathbf{v}|\binom{|S(\mathbf{v})|-1}{|\mathbf{w}|-1}P(\mathbf{w})}
\end{split}
\end{equation}

Before proceeding to a proof of the correctness of Algorithm 1, we stop to note a lemma that we will use in this and other proofs throughout the paper.

\begin{lemma}
\label{sumlemma}
\[\mathrm{Cost}\big(T_{S(\mathbf{v})}\big)=\sum_{\mathbf{w}\in H(\mathbf{v})}\frac{\mathrm{Cost}(T_\mathbf{w})}{\binom{|S(\mathbf{v})|-1}{|\mathbf{w}|-1}}\]
\end{lemma}
\begin{proof}
We begin by expanding the right hand side of the proposed equation.
\[\sum_{\mathbf{w}\in H(\mathbf{v})}\frac{\mathrm{Cost}(T_\mathbf{w})}{\binom{|S(\mathbf{v})|-1}{|\mathbf{w}|-1}}=\sum_{\mathbf{w}\in H(\mathbf{v})}\frac{1}{\binom{|S(\mathbf{v})|-1}{|\mathbf{w}|-1}}\sum_{w\in\mathbf{w}}\mathrm{Cost}(T_w)\]

Since $|\mathbf{w}|=\mathrm{min}(B,|S(\mathbf{v})|)$ does not depend on the particular choice of $\mathbf{w}$, we can move the factor in which it appears outside the summation.
\begin{equation*}
\begin{split}
\sum_{\mathbf{w}\in H(\mathbf{v})}\frac{\mathrm{Cost}(T_\mathbf{w})}{\binom{|S(\mathbf{v})|-1}{|\mathbf{w}|-1}}
&=\frac{1}{\binom{|S(\mathbf{v})|-1}{|\mathbf{w}|-1}}\sum_{\mathbf{w}\in H(\mathbf{v})}\sum_{w\in\mathbf{w}}\mathrm{Cost}(T_w)\\
\end{split}
\end{equation*}

Each $w\in S(\mathbf{v})$ appears in precisely $\binom{|S(\mathbf{v})|-1}{|\mathbf{w}|-1}$ of the $\mathbf{w}\in H(\mathbf{v})$, therefore we can simplify the double summation.
\begin{equation*}
\begin{split}
\sum_{\mathbf{w}\in H(\mathbf{v})}\frac{\mathrm{Cost}(T_\mathbf{w})}{\binom{|S(\mathbf{v})|-1}{|\mathbf{w}|-1}}
&=\frac{1}{\binom{|S(\mathbf{v})|-1}{|\mathbf{w}|-1}}\binom{|S(\mathbf{v})|-1}{|\mathbf{w}|-1}\sum_{w\in S(\mathbf{v})}\mathrm{Cost}(T_w)\\
&=\sum_{w\in S(\mathbf{v})}\mathrm{Cost}(T_w)\\
&=\mathrm{Cost}(T_{S(\mathbf{v})})
\end{split}
\end{equation*}\qed
\end{proof}

\begin{theorem}
\label{unbiased1}
Algorithm 1 is an unbiased estimator of tree cost, meaning
\[\mathbb{E}\big[C_{\mathrm{SEP}}(T_{\mathbf{v}})\big]=\frac{\mathrm{Cost}(T_\mathbf{v})}{|\mathbf{v}|}\]
\end{theorem}
\begin{proof}
The proof proceeds by induction over the height of the tree.  For a forest of height 0, we have $|S(\mathbf{v})|=0$, so the algorithm returns the exact answer
\[\frac{c(\mathbf{v})}{|\mathbf{v}|} =\frac{\mathrm{Cost}(T_\mathbf{v})}{|\mathbf{v}|} \]

Assuming that the proposition is correct for forests with heights strictly less than the height of $T_\mathbf{v}$, we have
\begin{equation*}
\begin{split}
\mathbb{E}[C_{\mathrm{SEP}}(T_{\mathbf{v}})]
&= \mathbb{E}\left[\frac{c(\mathbf{v})}{|\mathbf{v}|}+\frac{|\mathbf{w}|C_{\mathrm{SEP}}(T_{\mathbf{w}})}{|\mathbf{v}|\binom{|S(\mathbf{v})|-1}{|\mathbf{w}|-1}P(\mathbf{w})}\right]\\
&= \frac{c(\mathbf{v})}{|\mathbf{v}|}+\mathbb{E}\left[\frac{|\mathbf{w}|C_{\mathrm{SEP}}(T_{\mathbf{w}})}{|\mathbf{v}|\binom{|S(\mathbf{v})|-1}{|\mathbf{w}|-1}P(\mathbf{w})}\right]\\
&= \frac{c(\mathbf{v})}{|\mathbf{v}|}+\sum_{\mathbf{w}\in H(\mathbf{v})}P(\mathbf{w})\frac{|\mathbf{w}|\mathbb{E}[C_{\mathrm{SEP}}(T_{\mathbf{w}})]}{|\mathbf{v}|\binom{|S(\mathbf{v})|-1}{|\mathbf{w}|-1}P(\mathbf{w})}\\
&= \frac{c(\mathbf{v})}{|\mathbf{v}|}+\sum_{\mathbf{w}\in H(\mathbf{v})}\frac{|\mathbf{w}|}{|\mathbf{v}|\binom{|S(\mathbf{v})|-1}{|\mathbf{w}|-1}}\mathbb{E}[C_{\mathrm{SEP}}(T_{\mathbf{w}})]\\
&= \frac{c(\mathbf{v})}{|\mathbf{v}|}+\sum_{\mathbf{w}\in H(\mathbf{v})}\frac{|\mathbf{w}|}{|\mathbf{v}|\binom{|S(\mathbf{v})|-1}{|\mathbf{w}|-1}}\frac{\mathrm{Cost}(T_\mathbf{w})}{|\mathbf{w}|}\\
&= \frac{c(\mathbf{v})}{|\mathbf{v}|}+\frac{1}{|\mathbf{v}|}\sum_{\mathbf{w}\in H(\mathbf{v})}\frac{\mathrm{Cost}(T_\mathbf{w})}{\binom{|S(\mathbf{v})|-1}{|\mathbf{w}|-1}}\\
\end{split}
\end{equation*}

Applying Lemma \ref{sumlemma}, we get
\[\mathbb{E}[C_{\mathrm{SEP}}(T_{\mathbf{v}})]
= \frac{c(\mathbf{v})}{|\mathbf{v}|}+\frac{\mathrm{Cost}(T_{S(\mathbf{v})})}{|\mathbf{v}|}
= \frac{\mathrm{Cost}(T_\mathbf{v})}{|\mathbf{v}|}\]
\qed\end{proof}

Now that we know Algorithm 1 works, we can start thinking about how to improve the variance of the estimates it produces.

The purpose of using a non-uniform probability distribution to select each hypernode is to try to achieve a better variance between the estimates.  Therefore, it is important to know the optimal probability distribution, in other words, the probability distribution that would yield the exact answer for every estimate.

As with Knuth's algorithm, it turns out that the optimal probability for choosing a hypernode is proportional to the cost of the forest rooted at the hypernode.  Details are given below.

\begin{theorem}\label{idealdistro}
In Algorithm 1, if each hypernode $\mathbf{w}$ is chosen from all possible hypernodes in $H(\mathbf{v})$ with probability
\[P(\mathbf{w})=\frac{\mathrm{Cost}(T_\mathbf{w})}{\sum\limits_{\mathbf{x}\in H(\mathbf{v})}{\mathrm{Cost}(T_\mathbf{x})}}\]
then $C_{\mathrm{SEP}}$ is a zero-variance estimator, meaning
\[C_{\mathrm{SEP}}(T_{\mathbf{v}})=\frac{\mathrm{Cost}(T_\mathbf{v})}{|\mathbf{v}|}\]
\end{theorem}
\begin{proof}
The proof proceeds by induction over the height of the tree.  For a tree of height 0, we have $|S(\mathbf{v})|=0$, so the algorithm returns the exact answer
\[C_{\mathrm{SEP}}(T_{\mathbf{v}})=\frac{c(\mathbf{v})}{|\mathbf{v}|} =\frac{\mathrm{Cost}(T_\mathbf{v})}{|\mathbf{v}|} \]

Assuming that the proposition is correct for forests with heights strictly less than the height of $T_\mathbf{v}$, we have
\begin{equation*}
\begin{split}
C_{\mathrm{SEP}}(T_{\mathbf{v}})
&=\frac{c(\mathbf{v})}{|\mathbf{v}|}+\frac{|\mathbf{w}|C_{\mathrm{SEP}}(T_{\mathbf{w}})}{|\mathbf{v}|\binom{|S(\mathbf{v})|-1}{|\mathbf{w}|-1}P(\mathbf{w})}\\
&=\frac{c(\mathbf{v})}{|\mathbf{v}|}+\frac{\mathrm{Cost}(T_\mathbf{w})}{|\mathbf{v}|\binom{|S(\mathbf{v})|-1}{|\mathbf{w}|-1}P(\mathbf{w})}\\
&=\frac{c(\mathbf{v})}{|\mathbf{v}|}+\frac{\mathrm{Cost}(T_\mathbf{w})}{|\mathbf{v}|\binom{|S(\mathbf{v})|-1}{|\mathbf{w}|-1}\mathrm{Cost}(T_\mathbf{w})}\sum\limits_{\mathbf{x}\in H(\mathbf{v})}{\mathrm{Cost}(T_\mathbf{x})}\\
&=\frac{c(\mathbf{v})}{|\mathbf{v}|}+\frac{1}{|\mathbf{v}|\binom{|S(\mathbf{v})|-1}{|\mathbf{w}|-1}}\sum\limits_{\mathbf{x}\in H(\mathbf{v})}{\mathrm{Cost}(T_\mathbf{x})}\\
&=\frac{c(\mathbf{v})}{|\mathbf{v}|}+\frac{1}{|\mathbf{v}|}\sum\limits_{\mathbf{x}\in H(\mathbf{v})}{\frac{\mathrm{Cost}(T_\mathbf{x})}{\binom{|S(\mathbf{v})|-1}{|\mathbf{w}|-1}}}\\
\end{split}
\end{equation*}

Applying Lemma \ref{sumlemma}, we get
\[C_{\mathrm{SEP}}(T_{\mathbf{v}})= \frac{c(\mathbf{v})}{|\mathbf{v}|}+\frac{\mathrm{Cost}(T_{S(\mathbf{v})})}{|\mathbf{v}|}
= \frac{\mathrm{Cost}(T_\mathbf{v})}{|\mathbf{v}|}\]
\qed\end{proof}

We are now ready to discuss using an importance function to implement a probability distribution.

\section{Stochastic Enumeration with Importance}

The information in Theorem \ref{idealdistro} suggests that we should use a probability distribution in which each hypernode has a probability that is proportional to the cost of the forest beginning at that hypernode.  Obviously this will be difficult to achieve even as an estimate, since it is the same problem that we are trying to address with our algorithms.

However, even supposing that we did have some way of estimating the ideal probability for each hypernode, there is another problem with trying to implement a non-uniform probability distribution on the hypernodes.  Simply put, $|H(\mathbf{v})|=\binom{|S(\mathbf{v})|}{|\mathbf{w}|}$ may be extremely large, and so, if we hope to keep the running time of the algorithm under control, we need a way of choosing hypernodes that does not require us to calculate or store the probability of each individual hypernode in $H(\mathbf{v})$.

It turns out that there is an easy way to do this.  Consider a function $r$ from the nodes of a tree to the positive real numbers.  For a node $x$, we will call $r(x)$ the \textit{weight} of $x$ or the \textit{importance} of $x$.  We can extend the domain of $r$ to sets of nodes by defining the weight of a set of nodes $\mathbf{x}=\{x_1,x_2,\ldots,x_m\}$ as $r(\mathbf{x})=r(x_1)+r(x_2)+\cdots+r(x_m)$.

Given this weighting scheme, there is a way to choose a hypernode $\mathbf{w}$ with probability
\[P(\mathbf{w})=\frac{r(\mathbf{w})}{\sum\limits_{\mathbf{x}\in H(\mathbf{v})}{r(\mathbf{x})}}\]
that only requires us to calculate the weights of $S(\mathbf{v})$, and not of $H(\mathbf{v})$.  This method is described in Algorithm 2.

\begin{algorithm}[H]
\DontPrintSemicolon
\SetKwInOut{Input}{Input}\SetKwInOut{Output}{Output}
\Input{A forest $T_\mathbf{v}$ of height $h$ rooted at a hypernode $\mathbf{v}$, a budget $B\in\mathbb{N}$, and an importance function $r$}
\Output{An unbiased estimator $|\mathbf{v}|C_{\mathrm{SEI}}$ of the total cost of the forest $T_\mathbf{v}$}
\BlankLine
\textbf{(Initialization):} Set $k\leftarrow 0$, $D\leftarrow 1$, $\mathbf{x}_0=\mathbf{v}$, and $C_{\mathrm{SEI}}\leftarrow c(\mathbf{x}_0)/|\mathbf{x}_0|$.\;
\textbf{(Compute the successors):} Let $S(\mathbf{x}_k)$ be the set of all successors of $\mathbf{x}_k$.\;
\textbf{(Terminal position?):} If $|S(\mathbf{x}_k)|=0$, the algorithm stops, returning $|\mathbf{v}|C_{\mathrm{SEI}}$ as an estimator of $\mathrm{Cost}(T_\mathbf{v})$.\;
\textbf{(Advance):} Choose hypernode $\mathbf{x}_{k+1}\in H(\mathbf{x}_k)$ by first selecting $x\in S(\mathbf{x}_k)$ with probability ${r(x)}/{r(S(\mathbf{x}_k))}$ and then selecting the remaining elements of $\mathbf{x}_{k+1}$ uniformly at random from $S(\mathbf{x}_k)\setminus \{x\}$.\;
\textbf{(Update):} Set $D_k\leftarrow \frac{|\mathbf{x}_{k+1}|}{|\mathbf{x}_k|}\frac{r(S(\mathbf{x}_k))}{r(\mathbf{x}_{k+1})}$, set $D\leftarrow D\cdot D_k$, and set $C_{\mathrm{SEI}}\leftarrow C_{\mathrm{SEI}}+\frac{c(\mathbf{x}_{k+1})}{|\mathbf{x}_{k+1}|}D$.\;
\textbf{(Loop):} Increase $k$ by 1 and return to Step 2.\;
\caption{Stochastic Enumeration with importance sampling (SEI) algorithm for estimating the cost of a backtrack tree}
\end{algorithm}
\bigskip
It may not be obvious, but Algorithm 2 is simply Algorithm 1 with a specific probability distribution implemented, as we shall prove now.

\begin{theorem}\label{alg2unbiased}
Algorithm 2 is an unbiased estimator of tree cost, meaning
\[\mathbb{E}[C_{\mathrm{SEI}}(T_{\mathbf{v}})]=\frac{\mathrm{Cost}(T_\mathbf{v})}{|\mathbf{v}|}\]
\end{theorem}
\begin{proof}
We begin by calculating the probability with which each $\mathbf{x}_{k+1}$ is being selected.

Since one element, $x$, is selected separately from the rest of $\mathbf{x}_{k+1}$, there are $|\mathbf{x}_{k+1}|$ different and mutually exclusive ways in which we can get the same $\mathbf{x}_{k+1}$.  This is because each element in $\mathbf{x}_{k+1}$ can play the role of $x$.

Once an $x$ has been selected from $S(\mathbf{x}_k)$ with probability ${r(x)}/{r(S(\mathbf{x}_k))}$, the rest of the elements are selected uniformly at random from the remaining elements in $S(\mathbf{x}_k)$, so the remaining elements are collectively selected with probability $1/\binom{|S(\mathbf{x}_k)|-1}{|\mathbf{x}_{k+1}|-1}$.

Therefore the probability with which any given $\mathbf{x}_{k+1}$ is selected is
\[P(\mathbf{x}_{k+1})=\sum_{x\in\mathbf{x}_{k+1}}\frac{r(x)}{r(S(\mathbf{x}_k))}\frac{1}{\binom{|S(\mathbf{x}_k)|-1}{|\mathbf{x}_{k+1}|-1}}=\frac{r(\mathbf{x}_{k+1})}{r(S(\mathbf{x}_k))}\frac{1}{\binom{|S(\mathbf{x}_k)|-1}{|\mathbf{x}_{k+1}|-1}}\]

The formula for $D_k$ in Algorithm 2 is then obtained by a simple substitution into the formula given in Algorithm 1, and so the proposition follows from Theorem \ref{unbiased1}.
\qed\end{proof}

Let $\mathbf{w}$ be selected from $H(\mathbf{v})$ as described in Algorithm 2.  Then the probability with which $\mathbf{w}$ is selected is
\begin{equation*}
\begin{split}
P(\mathbf{w})
&=\frac{r(\mathbf{w})}{r(S(\mathbf{v}))}\frac{1}{\binom{|S(\mathbf{v})|-1}{|\mathbf{w}|-1}}\\
&=\frac{r(\mathbf{w})}{\sum\limits_{x\in S(\mathbf{v})}r(x)\binom{|S(\mathbf{v})|-1}{|\mathbf{w}|-1}}\\
\end{split}
\end{equation*}

Since each $x\in S(\mathbf{v})$ appears in precisely $\binom{|S(\mathbf{v})|-1}{|\mathbf{w}|-1}$ of the $\mathbf{x}\in H(\mathbf{v})$, we can also write this as
\[P(\mathbf{w})=\frac{r(\mathbf{w})}{\sum\limits_{\mathbf{x}\in H(\mathbf{v})}{r(\mathbf{x})}}\]
which was the desired probability.  Clearly, from Theorem \ref{idealdistro}, the ideal importance function would be $r(\mathbf{w})=\mathrm{Cost}(T_\mathbf{w})$.

Before analyzing this algorithm any further, let's look at an example to get a better idea of how it works.  

\begin{example}
\label{SEIexample}
Consider the tree in Figure \ref{SEIexamplefigure}, which is the same as that in the previous examples, except that it has been labeled with importance function values in addition to the names of the nodes.

To keep things simple, we are reusing as many parameters as possible from Example \ref{SEPexample}, so the budget is $B=2$ and the cost function $c$ is 1 on every node.  Again, the total cost of the tree is the number of nodes, 14, and this choice simplifies $\frac{c(\mathbf{x}_{k+1})}{|\mathbf{x}_{k+1}|}$ to 1, so the update command for $C_\mathrm{SEI}$ becomes 
\[C_{\mathrm{SEI}}\leftarrow C_{\mathrm{SEI}}+D\]

The importance function we are using for each node $x$ is the number of leaves under $x$, including $x$ itself if it is a leaf.  We have labeled the importance of each node after the node's name in the figure.

Now let's examine a possible sequence of hypernodes produced by Algorithm 2, as shown in the figure.

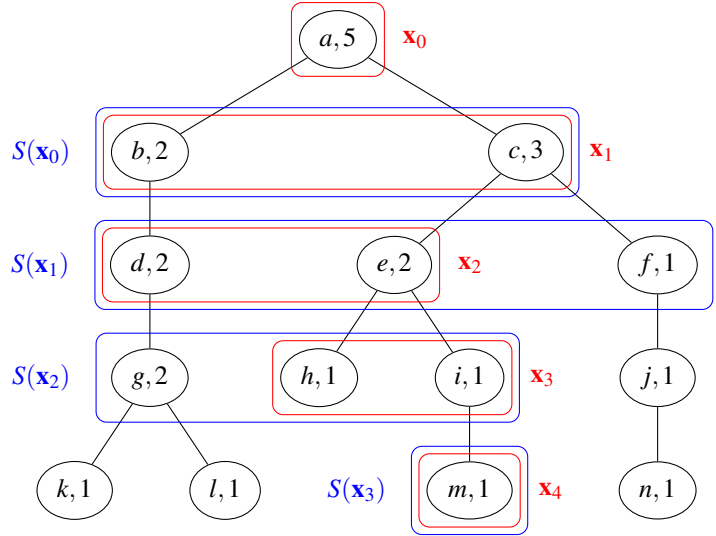
\begin{figure}
\centering
\begin{tikzpicture}[
level distance = 15mm,
level 1/.style={sibling distance=50mm},
level 2/.style={sibling distance=35mm},
level 3/.style={sibling distance=20mm},
every node/.style={draw,ellipse, minimum size=5mm}]
\node(a) {$a, 5$}
child {
	node(b) {$b,2$}
	child {
		node(d) {$d, 2$}
		child {
			node(g) {$g, 2$}
			child {
				node(k) {$k, 1$}
				}
			child {node(l) {$l, 1$}}
			}
		}
	}
child {
	node(c){$c, 3$}
	child {
		node(e) {$e, 2$}
		child {	node(h) {$h, 1$}}		
		child {
			node(i) {$i, 1$}
			child {node(m) {$m, 1$}}
			}
		}
	child {
		node(f) {$f, 1$}
		child {
			node(j) {$j, 1$}
			child { node(n) {$n, 1$} }
			}
		}
	};
\node[draw=red, fit=(a),inner sep=1mm,rectangle,rounded corners,label={[red]right:$\mathbf{x}_0$}](x0){};
\node[draw=red, fit=(b) (c),inner sep=1mm,rectangle,rounded corners,label={[red]right:$\mathbf{x}_1$}](x1){};
\node[draw=red, fit=(d) (e),inner sep=1mm,rectangle,rounded corners,label={[red]right:$\mathbf{x}_2$}](x2){};
\node[draw=red, fit=(h) (i),inner sep=1mm,rectangle,rounded corners,label={[red]right:$\mathbf{x}_3$}](x3){};
\node[draw=red, fit=(m),inner sep=1mm,rectangle,rounded corners,label={[red]right:$\mathbf{x}_4$}](x4){};

\node[draw=blue, fit=(b) (c),inner sep=2mm,rectangle,rounded corners,label={[blue]left:$S(\mathbf{x}_0)$}](Sx0){};
\node[draw=blue, fit=(d) (e) (f),inner sep=2mm,rectangle,rounded corners,label={[blue]left:$S(\mathbf{x}_1)$}](Sx1){};
\node[draw=blue, fit=(g) (h) (i),inner sep=2mm,rectangle,rounded corners,label={[blue]left:$S(\mathbf{x}_2)$}](Sx2){};
\node[draw=blue, fit=(m),inner sep=2mm,rectangle,rounded corners,label={[blue]left:$S(\mathbf{x}_3)$}](Sx3){};
\end{tikzpicture}
\caption{Tree for Example \ref{SEIexample}, with each chosen hypernode boxed and labeled to the right, and each chosen hypernode's successor set boxed and labeled to the left}
\label{SEIexamplefigure}
\end{figure}

We initialize with $k=0$, $\mathbf{x}_0=\{a\}$, $D=1$, $C_{\mathrm{SEI}}=1$.  Then we compute $S(\mathbf{x}_0)=\{b,c\}$.  We choose $c$ with probability
\[P(c)=\frac{r(c)}{r(S(\mathbf{x}_0))}=\frac{3}{2+3}=\frac{3}{5}\]
and then choose $b$ uniformly at random from the remaining elements, to give us $\mathbf{x}_1=\{b,c\}$.  We update
\[D_0\leftarrow \frac{|\mathbf{x}_{1}|}{|\mathbf{x}_0|}\frac{r(S(\mathbf{x}_0))}{r(\mathbf{x}_{1})}=\frac{2}{1}\cdot\frac{2+3}{2+3}=2\]
\[D\leftarrow D\cdot D_0=2\]
\[C_{\mathrm{SEI}}\leftarrow C_{\mathrm{SEI}}+D=3\]

We increase to $k=1$ and loop.  We compute $S(\mathbf{x}_1)=\{d,e,f\}$.  We choose $e$ with probability
\[P(e)=\frac{r(e)}{r(S(\mathbf{x}_1))}=\frac{2}{2+2+1}=\frac{2}{5}\]
and then choose $d$ uniformly at random from the remaining elements, giving us $\mathbf{x}_2=\{d,e\}$.  We update
\[D_1\leftarrow \frac{|\mathbf{x}_{2}|}{|\mathbf{x}_1|}\frac{r(S(\mathbf{x}_1))}{r(\mathbf{x}_{2})}=\frac{2}{2}\cdot\frac{2+2+1}{2+2}=\frac{5}{4}\]
\[D\leftarrow D\cdot D_1=\frac{5}{2}\]
\[C_{\mathrm{SEI}}\leftarrow C_{\mathrm{SEI}}+D=\frac{11}{2}\]

We increase to $k=2$ and loop.  We compute $S(\mathbf{x}_2)=\{g,h,i\}$.  We choose $i$ with probability
\[P(i)=\frac{r(i)}{r(S(\mathbf{x}_2))}=\frac{1}{2+1+1}=\frac{1}{4}\]
and then choose $h$ uniformly at random from the remaining elements, giving us $\mathbf{x}_2=\{h,i\}$.  We update
\[D_2\leftarrow \frac{|\mathbf{x}_{3}|}{|\mathbf{x}_2|}\frac{r(S(\mathbf{x}_2))}{r(\mathbf{x}_{3})}=\frac{2}{2}\cdot\frac{2+1+1}{1+1}=2\]
\[D\leftarrow D\cdot D_2=5\]
\[C_{\mathrm{SEI}}\leftarrow C_{\mathrm{SEI}}+D=\frac{21}{2}\]

We increase to $k=3$ and loop.  We compute $S(\mathbf{x}_3)=\{m\}$, and we choose $m$ with probability
\[P(m)=\frac{r(m)}{r(S(\mathbf{x}_3))}=\frac{1}{1}=1\]
Since there are no remaining elements to be chosen, we have $\mathbf{x}_4=\{m\}$.  We then update
\[D_3\leftarrow \frac{|\mathbf{x}_{4}|}{|\mathbf{x}_3|}\frac{r(S(\mathbf{x}_3))}{r(\mathbf{x}_{4})}=\frac{1}{2}\frac{1}{1}=\frac{1}{2}\]
\[D\leftarrow D\cdot D_3=\frac{5}{2}\]
\[C_{\mathrm{SEI}}\leftarrow C_{\mathrm{SEI}}+D=\frac{26}{2}=13\]

We increase to $k=4$ and loop.  We compute $S(\mathbf{x}_4)=\emptyset$, so we are in the terminal position and we stop.  The algorithm returns $|\mathbf{x}_0|C_{\mathrm{SEP}}=13$ as an estimator of the cost of the tree.  This completes the example.\qed
\end{example}

\section{Variance}

Recall that in Equation \ref{recursiveoutputeqn}, we found a recursive expression for the output of Algorithm 1 as

\[C_{\mathrm{SEP}}(T_{\mathbf{v}})=\frac{c(\mathbf{v})}{|\mathbf{v}|}+\frac{|\mathbf{w}|C_{\mathrm{SEP}}(T_{\mathbf{w}})}{|\mathbf{v}|\binom{|S(\mathbf{v})|-1}{|\mathbf{w}|-1}P(\mathbf{w})}\]

By substituting for $P(\mathbf{w})$ with the expression we found in the proof of Theorem \ref{alg2unbiased}, we get another recursive formula for the output of Algorithm 2.
\[C_{\mathrm{SEI}}(T_{\mathbf{v}})
=\frac{c(\mathbf{v})}{|\mathbf{v}|}+\frac{|\mathbf{w}|}{|\mathbf{v}|}\frac{r(S(\mathbf{v}))}{r(\mathbf{w})}C_{\mathrm{SEI}}(T_{\mathbf{w}})\]

With this information we can begin to analyze the variance of $C_{\mathrm{SEI}}$, or rather, the variance of $|\mathbf{v}|C_{\mathrm{SEI}}$, which is the actual estimate of tree cost produced by Algorithm 2.

\begin{theorem}
\label{alg2var}
For a forest $T_\mathbf{v}$ rooted at a hypernode $\mathbf{v}$, the variance produced by Algorithm 2 is
\begin{multline*}
\mathrm{Var}\big(|\mathbf{v}|C_{\mathrm{SEI}}(T_\mathbf{v})\big)\\
=\sum_{\mathbf{w}\in H(\mathbf{v})}\frac{1}{\binom{|S(\mathbf{v})|-1}{|\mathbf{w}|-1}}\frac{r\big(S(\mathbf{v})\big)}{r(\mathbf{w})}\Big(\mathrm{Var}\big(|\mathbf{w}|C_{\mathrm{SEI}}(T_\mathbf{w})\big)+\mathrm{Cost}(T_\mathbf{w})^2\Big)\\
-\mathrm{Cost}(T_{S(\mathbf{v})})^2
\end{multline*}
\end{theorem}
\begin{proof}
We know
\[C_{\mathrm{SEI}}(T_{\mathbf{v}})
=\frac{c(\mathbf{v})}{|\mathbf{v}|}+\frac{|\mathbf{w}|}{|\mathbf{v}|}\frac{r(S(\mathbf{v}))}{r(\mathbf{w})}C_{\mathrm{SEI}}(T_{\mathbf{w}})\]
which implies
\[|\mathbf{v}|C_{\mathrm{SEI}}(T_{\mathbf{v}})
=c(\mathbf{v})+\frac{r(S(\mathbf{v}))}{r(\mathbf{w})}|\mathbf{w}|C_{\mathrm{SEI}}(T_{\mathbf{w}})\]

Taking the variance of both sides, we get
\begin{equation*}
\begin{split}
\mathrm{Var}\big(|\mathbf{v}|C_{\mathrm{SEI}}(T_\mathbf{v})\big)
&=\mathrm{Var}\left(c(\mathbf{v})+\frac{r(S(\mathbf{v}))}{r(\mathbf{w})}|\mathbf{w}|C_{\mathrm{SEI}}(T_{\mathbf{w}})\right)\\
&=\mathrm{Var}\left(\frac{r(S(\mathbf{v}))}{r(\mathbf{w})}|\mathbf{w}|C_{\mathrm{SEI}}(T_{\mathbf{w}})\right)
\end{split}
\end{equation*}
and so
\begin{multline}\label{bigvariance}
\mathrm{Var}\big(|\mathbf{v}|C_{\mathrm{SEI}}(T_\mathbf{v})\big)\\
=\mathbb{E}\left[\left(\frac{r(S(\mathbf{v}))}{r(\mathbf{w})}|\mathbf{w}|C_{\mathrm{SEI}}(T_{\mathbf{w}})\right)^2\right]-\left(\mathbb{E}\left[\frac{r(S(\mathbf{v}))}{r(\mathbf{w})}|\mathbf{w}|C_{\mathrm{SEI}}(T_{\mathbf{w}})\right]\right)^2
\end{multline}

We will tackle each of these terms separately.  First,
\begin{multline}\label{variance1}
\mathbb{E}\left[\left(\frac{r(S(\mathbf{v}))}{r(\mathbf{w})}|\mathbf{w}|C_{\mathrm{SEI}}(T_{\mathbf{w}})\right)^2\right]\\
\begin{aligned}
&=\sum_{\mathbf{w}\in H(\mathbf{v})}P(\mathbf{w})\left(\frac{r(S(\mathbf{v}))}{r(\mathbf{w})}\right)^2\mathbb{E}\left[\left(|\mathbf{w}|C_{\mathrm{SEI}}(T_{\mathbf{w}})\right)^2\right]\\
&=\sum_{\mathbf{w}\in H(\mathbf{v})}\frac{1}{\binom{|S(\mathbf{v})|-1}{|\mathbf{w}|-1}}\frac{r(\mathbf{w})}{r(S(\mathbf{v}))}\left(\frac{r(S(\mathbf{v}))}{r(\mathbf{w})}\right)^2\mathbb{E}\left[\left(|\mathbf{w}|C_{\mathrm{SEI}}(T_{\mathbf{w}})\right)^2\right]\\
&=\sum_{\mathbf{w}\in H(\mathbf{v})}\frac{1}{\binom{|S(\mathbf{v})|-1}{|\mathbf{w}|-1}}\frac{r(S(\mathbf{v}))}{r(\mathbf{w})}\mathbb{E}\left[\left(|\mathbf{w}|C_{\mathrm{SEI}}(T_{\mathbf{w}})\right)^2\right]\\
&=\sum_{\mathbf{w}\in H(\mathbf{v})}\frac{1}{\binom{|S(\mathbf{v})|-1}{|\mathbf{w}|-1}}\frac{r(S(\mathbf{v}))}{r(\mathbf{w})}\Big(\mathrm{Var}\big(|\mathbf{w}|C_{\mathrm{SEI}}(T_\mathbf{w})\big)+\big(\mathbb{E}\left[|\mathbf{w}|C_{\mathrm{SEI}}(T_{\mathbf{w}})\right]\big)^2\Big)\\
&=\sum_{\mathbf{w}\in H(\mathbf{v})}\frac{1}{\binom{|S(\mathbf{v})|-1}{|\mathbf{w}|-1}}\frac{r(S(\mathbf{v}))}{r(\mathbf{w})}\Big(\mathrm{Var}\big(|\mathbf{w}|C_{\mathrm{SEI}}(T_\mathbf{w})\big)+\mathrm{Cost}(T_\mathbf{w})^2\Big)
\end{aligned}
\end{multline}

Next,
\begin{multline}\label{variance2}
\mathbb{E}\left[\frac{r(S(\mathbf{v}))}{r(\mathbf{w})}|\mathbf{w}|C_{\mathrm{SEI}}(T_{\mathbf{w}})\right]\\
\begin{aligned}
&=\sum_{\mathbf{w}\in H(\mathbf{v})}P(\mathbf{w})\frac{r(S(\mathbf{v}))}{r(\mathbf{w})}\mathbb{E}\left[|\mathbf{w}|C_{\mathrm{SEI}}(T_{\mathbf{w}})\right]\\
&=\sum_{\mathbf{w}\in H(\mathbf{v})}\frac{1}{\binom{|S(\mathbf{v})|-1}{|\mathbf{w}|-1}}\frac{r(\mathbf{w})}{r(S(\mathbf{v}))}\frac{r(S(\mathbf{v}))}{r(\mathbf{w})}\mathbb{E}\left[|\mathbf{w}|C_{\mathrm{SEI}}(T_{\mathbf{w}})\right]\\
&=\sum_{\mathbf{w}\in H(\mathbf{v})}\frac{1}{\binom{|S(\mathbf{v})|-1}{|\mathbf{w}|-1}}\mathbb{E}\left[|\mathbf{w}|C_{\mathrm{SEI}}(T_{\mathbf{w}})\right]\\
&=\sum_{\mathbf{w}\in H(\mathbf{v})}\frac{\mathrm{Cost}(T_\mathbf{w})}{\binom{|S(\mathbf{v})|-1}{|\mathbf{w}|-1}}\\
&=\mathrm{Cost}(T_{S(\mathbf{v})})
\end{aligned}
\end{multline}
where the last line is due to Lemma \ref{sumlemma}.

Substituting the results of Equations \ref{variance1} and \ref{variance2} into \ref{bigvariance} yields the theorem statement.
\qed\end{proof}

From Theorem \ref{alg2var} we can easily find an expression for the coefficient of variation (CV).

\begin{corollary}
\label{CVcor}
For a forest $T_\mathbf{v}$ rooted at a hypernode $\mathbf{v}$, the coefficient of variation produced by Algorithm 2 is given by
\begin{multline*}
\mathrm{CV}^2\big(|\mathbf{v}|C_{\mathrm{SEI}}(T_\mathbf{v})\big)\\
=\sum_{\mathbf{w}\in H(\mathbf{v})}\frac{1}{\binom{|S(\mathbf{v})|-1}{|\mathbf{w}|-1}}\frac{r\big(S(\mathbf{v})\big)}{r(\mathbf{w})}\left(\frac{\mathrm{Cost}(T_\mathbf{w})}{\mathrm{Cost}(T_\mathbf{v})}\right)^2\Big(\mathrm{CV}^2\big(|\mathbf{w}|C_{\mathrm{SEI}}(T_\mathbf{w})\big)+1\Big)\\
-\left(\frac{\mathrm{Cost}(T_{S(\mathbf{v})})}{\mathrm{Cost}(T_\mathbf{v})}\right)^2
\end{multline*}
\end{corollary}
\begin{proof}
This follows from Theorem \ref{alg2var} by dividing both sides of that equation by $\big(\mathrm{Cost}(T_\mathbf{v})\big)^2$ and then factoring out $\big(\mathrm{Cost}(T_\mathbf{w})\big)^2$ from $\mathrm{Var}\big(|\mathbf{w}|C_{\mathrm{SEI}}(T_\mathbf{w})\big)+\mathrm{Cost}(T_\mathbf{w})^2$.
\qed\end{proof}

\section{Upper Bounds on the Variance}

The first bound we present is almost as complicated as the formula for the coefficient of variation given in the previous section; however, this bound is very useful for proving other bounds because of its recursive structure, and so we present it first.

\begin{theorem}
\label{alg2cv}
For a forest $T_\mathbf{v}$ rooted at a hypernode $\mathbf{v}$, the coefficient of variation produced by Algorithm 2 is bounded by
\begin{multline*}
\mathrm{CV}^2\big(|\mathbf{v}|C_{\mathrm{SEI}}(T_\mathbf{v})\big)+1\\
\leq\sum_{\mathbf{w}\in H(\mathbf{v})}\frac{1}{\binom{|S(\mathbf{v})|-1}{|\mathbf{w}|-1}}\frac{r\big(S(\mathbf{v})\big)}{r(\mathbf{w})}\left(\frac{\mathrm{Cost}(T_\mathbf{w})}{\mathrm{Cost}(T_{S(\mathbf{v})})}\right)^2\Big(\mathrm{CV}^2\big(|\mathbf{w}|C_{\mathrm{SEI}}(T_\mathbf{w})\big)+1\Big)
\end{multline*}
\end{theorem}
\begin{proof}
To save space we will refer to $\mathrm{CV}^2\big(|\mathbf{v}|C_{\mathrm{SEI}}(T_\mathbf{v})\big)+1$ as $F(\mathbf{v})$ and likewise for other hypernodes.  With this notation, Corollary \ref{CVcor} says that
\begin{multline*}
F(\mathbf{v})-1\\
=\left(\sum_{\mathbf{w}\in H(\mathbf{v})}\frac{1}{\binom{|S(\mathbf{v})|-1}{|\mathbf{w}|-1}}\frac{r\big(S(\mathbf{v})\big)}{r(\mathbf{w})}\left(\frac{\mathrm{Cost}(T_\mathbf{w})}{\mathrm{Cost}(T_\mathbf{v})}\right)^2F(\mathbf{w})\right)
-\left(\frac{\mathrm{Cost}(T_{S(\mathbf{v})})}{\mathrm{Cost}(T_\mathbf{v})}\right)^2
\end{multline*}

Now we factor the right hand side (moving the subtracted term before the sum to avoid confusion) and get
\begin{multline*}
F(\mathbf{v})-1=\\
\left(\frac{\mathrm{Cost}(T_{S(\mathbf{v})})}{\mathrm{Cost}(T_\mathbf{v})}\right)^2\left(-1+\sum_{\mathbf{w}\in H(\mathbf{v})}\frac{1}{\binom{|S(\mathbf{v})|-1}{|\mathbf{w}|-1}}\frac{r\big(S(\mathbf{v})\big)}{r(\mathbf{w})}\left(\frac{\mathrm{Cost}(T_\mathbf{w})}{\mathrm{Cost}(T_{S(\mathbf{v})})}\right)^2F(\mathbf{w})\right)
\end{multline*}

Since the first factor on the right hand side is less than 1, this implies
\[F(\mathbf{v})-1
\leq-1+\sum_{\mathbf{w}\in H(\mathbf{v})}\frac{1}{\binom{|S(\mathbf{v})|-1}{|\mathbf{w}|-1}}\frac{r\big(S(\mathbf{v})\big)}{r(\mathbf{w})}\left(\frac{\mathrm{Cost}(T_\mathbf{w})}{\mathrm{Cost}(T_{S(\mathbf{v})})}\right)^2F(\mathbf{w})\]

After canceling $-1$ from both sides, we have the result. 
\qed\end{proof}

Next, we will show we can bound the coefficient of variation by a measure of how close the importance function $r$ is to the Cost function.  First, we need a few new definitions.

For a tree $T_{\mathbf{v}}$ of height $n$, rooted at $\mathbf{v}$, define $\mathbf{x}_0:=\mathbf{v}$, and let $\mathbf{x}_0,\mathbf{x}_1,\dots,\mathbf{x}_n$  be any possible sequence of hypernodes that can be produced by Algorithm 2.  We define a function $\alpha$ on such sequences that describes how close our importance function is to the Cost function over the sequence.  Let $\alpha$ be given by
\[\alpha(\mathbf{x}_0,\mathbf{x}_1,\dots,\mathbf{x}_n)=\prod_{i=1}^n\frac{r(S(\mathbf{x}_{i-1}))}{r(\mathbf{x}_i)}\frac{\mathrm{Cost}(T_{\mathbf{x}_i})}{\mathrm{Cost}(T_{S(\mathbf{x}_{i-1})})}\]

Note that for a tree of height zero, the product is empty, and so $\alpha=1$.  We also have $\alpha=1$ in the case that we use the exact Cost function for the importance function.

Before stating our next bound, we first need a lemma regarding the expected value of $\alpha$.
\begin{lemma}
\label{alphalemma}
For a forest $T_{\mathbf{v}}$, let $\mathbf{v},\mathbf{x}_1,\dots,\mathbf{x}_n$ be any possible hypernode sequence produced by Algorithm 2.  Then
\[\mathbb{E}\big[\alpha(\mathbf{v},\mathbf{x}_1,\dots,\mathbf{x}_n)\big]=1\]
\end{lemma}
\begin{proof}
The proof proceeds by induction on the height of the tree.  For a tree $T_\mathbf{v}$ of height 0, $\alpha(\mathbf{v})=1$ always.

For the inductive step, let $T_\mathbf{v}$ be a forest of height $n$ and assume the proposition for all forests of heights $n-1$ or less.  Define $\mathbf{x}_0:=\mathbf{v}$. Then
\[
\mathbb{E}\big[\alpha(\mathbf{x}_0,\dots,\mathbf{x}_n)\big]=\sum_{\substack{\mathbf{x}_1,\dots,\mathbf{x}_n : \\ \mathbf{x}_{i+1}\in H(\mathbf{x}_i)}}P(\mathbf{x}_1,\dots,\mathbf{x}_n)\alpha(\mathbf{x}_0,\dots,\mathbf{x}_n)\]

The right hand side sum can be expanded to
\begin{multline*}
\sum_{\mathbf{x}_1\in H(\mathbf{v})}P(\mathbf{x}_1)\frac{r(S(\mathbf{v}))}{r(\mathbf{x}_1)}\frac{\mathrm{Cost}(T_{\mathbf{x}_1})}{\mathrm{Cost}(T_{S(\mathbf{v})})}\sum_{\substack{\mathbf{x}_2,\dots,\mathbf{x}_n : \\ \mathbf{x}_{i+1}\in H(\mathbf{x}_i)}}P(\mathbf{x}_2,\dots,\mathbf{x}_n \mid \mathbf{x}_1)\alpha(\mathbf{x}_1,\dots,\mathbf{x}_n)\\
\begin{aligned}
&=\sum_{\mathbf{x}_1\in H(\mathbf{v})}P(\mathbf{x}_1)\frac{r(S(\mathbf{v}))}{r(\mathbf{x}_1)}\frac{\mathrm{Cost}(T_{\mathbf{x}_1})}{\mathrm{Cost}(T_{S(\mathbf{v})})}\mathbb{E}\big[\alpha(\mathbf{x}_1,\dots,\mathbf{x}_n)\big]\\
&=\sum_{\mathbf{x}_1\in H(\mathbf{v})}P(\mathbf{x}_1)\frac{r(S(\mathbf{v}))}{r(\mathbf{x}_1)}\frac{\mathrm{Cost}(T_{\mathbf{x}_1})}{\mathrm{Cost}(T_{S(\mathbf{v})})}
\end{aligned}
\end{multline*}

Recall that the probability of choosing $\mathbf{x}_1$ from $H(\mathbf{v})$ is
\[P(\mathbf{x}_1)=\frac{1}{\binom{|S(\mathbf{v})|-1}{|\mathbf{x}_1|-1}}\frac{r(\mathbf{x}_1)}{r(S(\mathbf{v}))}\]

Hence
\begin{equation*}
\begin{split}
\mathbb{E}\big[\alpha(\mathbf{x}_0,\dots,\mathbf{x}_n)\big]
&=\sum_{\mathbf{x}_1\in H(\mathbf{v})}P(\mathbf{x}_1)\frac{r(S(\mathbf{v}))}{r(\mathbf{x}_1)}\frac{\mathrm{Cost}(T_{\mathbf{x}_1})}{\mathrm{Cost}(T_{S(\mathbf{v})})}\\
&=\sum_{\mathbf{x}_1\in H(\mathbf{v})}\frac{1}{\binom{|S(\mathbf{v})|-1}{|\mathbf{x}_1|-1}}\frac{r(\mathbf{x}_1)}{r(S(\mathbf{v}))}\frac{r(S(\mathbf{v}))}{r(\mathbf{x}_1)}\frac{\mathrm{Cost}(T_{\mathbf{x}_1})}{\mathrm{Cost}(T_{S(\mathbf{v})})}\\
&=\sum_{\mathbf{x}_1\in H(\mathbf{v})}\frac{1}{\binom{|S(\mathbf{v})|-1}{|\mathbf{x}_1|-1}}\frac{\mathrm{Cost}(T_{\mathbf{x}_1})}{\mathrm{Cost}(T_{S(\mathbf{v})})}\\
&=\frac{1}{\mathrm{Cost}(T_{S(\mathbf{v})})}\sum_{\mathbf{x}_1\in H(\mathbf{v})}\frac{\mathrm{Cost}(T_{\mathbf{x}_1})}{\binom{|S(\mathbf{v})|-1}{|\mathbf{x}_1|-1}}\\
&=1\\
\end{split}
\end{equation*}
where the last line is due to Lemma \ref{sumlemma}.
\qed\end{proof}

Now we are ready to state our next bound.

\begin{theorem}
For a forest $T_{\mathbf{v}}$, let $\mathbf{v},\mathbf{x}_1,\dots,\mathbf{x}_n$ be any possible hypernode sequence produced by Algorithm 2.  Then
\[\mathrm{CV}^2\big(|\mathbf{v}|C_{\mathrm{SEI}}(T_\mathbf{v})\big)\leq\mathrm{Var}(\alpha(\mathbf{v},\mathbf{x}_1,\dots,\mathbf{x}_n))\]
\end{theorem}
\begin{proof}
The proof proceeds by induction on the height of the tree.  For a tree $T_\mathbf{v}$ of height 0, $\alpha(\mathbf{v})=1$, and the algorithm is exact, so both sides of the inequality are zero.

For the inductive step, let $T_\mathbf{v}$ be a forest of height $n$ and assume the proposition for all forests of heights $n-1$ or less.  Define $\mathbf{x}_0:=\mathbf{v}$. Then Theorem \ref{alg2cv} says
\begin{multline*}
\mathrm{CV}^2\big(|\mathbf{x}_0|C_{\mathrm{SEI}}(T_{\mathbf{x}_0})\big)+1\\
\leq\sum_{\mathbf{x}_1\in H(\mathbf{x}_0)}\frac{1}{\binom{|S(\mathbf{x}_0)|-1}{|\mathbf{x}_1|-1}}\frac{r\big(S(\mathbf{x}_0)\big)}{r(\mathbf{x}_1)}\left(\frac{\mathrm{Cost}(T_{\mathbf{x}_1})}{\mathrm{Cost}(T_{S(\mathbf{x}_0)})}\right)^2\Big(\mathrm{CV}^2\big(|\mathbf{x}_1|C_{\mathrm{SEI}}(T_{\mathbf{x}_1})\big)+1\Big)
\end{multline*}

Recall that the probability of choosing $\mathbf{x}_1$ from $H(\mathbf{x}_0)$ is
\[P(\mathbf{x}_1)=\frac{1}{\binom{|S(\mathbf{x}_0)|-1}{|\mathbf{x}_1|-1}}\frac{r(\mathbf{x}_1)}{r(S(\mathbf{x}_0))}\]

This implies
\begin{equation}
\frac{1}{\binom{|S(\mathbf{x}_0)|-1}{|\mathbf{x}_1|-1}}=P(\mathbf{x}_1)\frac{r(S(\mathbf{x}_0))}{r(\mathbf{x}_1)}
\end{equation}

Substituting the right hand side of (5) for the left hand side of (5) in the inequality yields
\begin{multline*}
\mathrm{CV}^2\big(|\mathbf{x}_0|C_{\mathrm{SEI}}(T_{\mathbf{x}_0})\big)+1\\
\leq\sum_{\mathbf{x}_1\in H(\mathbf{x}_0)}P(\mathbf{x}_1)\left(\frac{r\big(S(\mathbf{x}_0)\big)}{r(\mathbf{x}_1)}\frac{\mathrm{Cost}(T_{\mathbf{x}_1})}{\mathrm{Cost}(T_{S(\mathbf{x}_0)})}\right)^2\Big(\mathrm{CV}^2\big(|\mathbf{x}_1|C_{\mathrm{SEI}}(T_{\mathbf{x}_1})\big)+1\Big)
\end{multline*}

Then from the induction hypothesis, this becomes
\begin{multline*}
\mathrm{CV}^2\big(|\mathbf{x}_0|C_{\mathrm{SEI}}(T_{\mathbf{x}_0})\big)+1\\
\leq\sum_{\mathbf{x}_1\in H(\mathbf{x}_0)}P(\mathbf{x}_1)\left(\frac{r\big(S(\mathbf{x}_0)\big)}{r(\mathbf{x}_1)}\frac{\mathrm{Cost}(T_{\mathbf{x}_1})}{\mathrm{Cost}(T_{S(\mathbf{x}_0)})}\right)^2\Big(\mathrm{Var}\big(\alpha(\mathbf{x}_1,\dots,\mathbf{x}_n) \mid \mathbf{x}_1\big)+1\Big)
\end{multline*}

Then due to Lemma \ref{alphalemma} we can write this as
\begin{multline*}
\mathrm{CV}^2\big(|\mathbf{x}_0|C_{\mathrm{SEI}}(T_{\mathbf{x}_0})\big)+1\\
\begin{aligned}
&\leq\sum_{\mathbf{x}_1\in H(\mathbf{x}_0)}P(\mathbf{x}_1)\left(\frac{r\big(S(\mathbf{x}_0)\big)}{r(\mathbf{x}_1)}\frac{\mathrm{Cost}(T_{\mathbf{x}_1})}{\mathrm{Cost}(T_{S(\mathbf{x}_0)})}\right)^2\mathbb{E}\big[\alpha(\mathbf{x}_1,\dots,\mathbf{x}_n)^2 \mid \mathbf{x}_1\big]\\
&=\sum_{\mathbf{x}_1\in H(\mathbf{x}_0)}P(\mathbf{x}_1)\mathbb{E}\big[\alpha(\mathbf{x}_0,\mathbf{x}_1,\dots,\mathbf{x}_n)^2 \mid \mathbf{x}_1\big]\\
&=\mathbb{E}\big[\alpha(\mathbf{x}_0,\mathbf{x}_1,\dots,\mathbf{x}_n)^2\big]\\
&=\mathrm{Var}\big(\alpha(\mathbf{x}_0,\mathbf{x}_1,\dots,\mathbf{x}_n)\big)+\mathbb{E}\big[\alpha(\mathbf{x}_0,\mathbf{x}_1,\dots,\mathbf{x}_n)\big]^2\\
&=\mathrm{Var}\big(\alpha(\mathbf{x}_0,\mathbf{x}_1,\dots,\mathbf{x}_n)\big)+1
\end{aligned}
\end{multline*}
\qed\end{proof}

Recognizing that it may be very difficult to calculate the variance of $\alpha$, we now introduce a looser bound which may be easier to calculate.

\begin{theorem}
\label{maxalpha}
For a forest $T_{\mathbf{v}}$, let $\mathcal{X}(\mathbf{v})$ be the collection of all possible hypernode sequences $(\mathbf{x}_1,\dots,\mathbf{x}_n)$ which can occur after $\mathbf{x}_0=\mathbf{v}$ in Algorithm 2.  Then
\[\mathrm{CV}^2\big(|\mathbf{v}|C_{\mathrm{SEI}}(T_\mathbf{v})\big)\leq\max_{\substack{(\mathbf{x}_1,\dots,\mathbf{x}_n) \\ \in\mathcal{X}(\mathbf{v})}}\alpha(\mathbf{v},\mathbf{x}_1,\dots,\mathbf{x}_n)-1\]
\end{theorem}
\begin{proof}
The proof proceeds by induction on the height of the tree.  For a tree $T_\mathbf{v}$ of height 0, we have $\alpha(\mathbf{v})=1$, so the right hand side is zero, and we know that the algorithm is exact on trees of height zero, so the left hand side is also zero.

For the inductive step, let $T_\mathbf{v}$ be a forest of height $n$ and assume the proposition for all forests of heights $n-1$ or less.  Define $\mathbf{x}_0:=\mathbf{v}$. Then Theorem \ref{alg2cv} and the inductive hypothesis give us
\begin{multline*}
\mathrm{CV}^2\big(|\mathbf{x}_0|C_{\mathrm{SEI}}(T_{\mathbf{x}_0})\big)+1\\
\begin{aligned}
&\leq\sum_{\mathbf{x}_1\in H(\mathbf{x}_0)}\frac{1}{\binom{|S(\mathbf{x}_0)|-1}{|\mathbf{x}_1|-1}}\frac{r\big(S(\mathbf{x}_0)\big)}{r(\mathbf{x}_1)}\left(\frac{\mathrm{Cost}(T_{\mathbf{x}_1})}{\mathrm{Cost}(T_{S(\mathbf{x}_0)})}\right)^2\Big(\mathrm{CV}^2\big(|\mathbf{x}_1|C_{\mathrm{SEI}}(T_{\mathbf{x}_1})\big)+1\Big)\\
&\leq\sum_{\mathbf{x}_1\in H(\mathbf{x}_0)}\frac{1}{\binom{|S(\mathbf{x}_0)|-1}{|\mathbf{x}_1|-1}}\frac{r\big(S(\mathbf{x}_0)\big)}{r(\mathbf{x}_1)}\left(\frac{\mathrm{Cost}(T_{\mathbf{x}_1})}{\mathrm{Cost}(T_{S(\mathbf{x}_0)})}\right)^2\max_{\substack{(\mathbf{x}_2,\dots,\mathbf{x}_n) \\ \in\mathcal{X}(\mathbf{x}_1)}}\alpha(\mathbf{x}_1,\mathbf{x}_2,\dots,\mathbf{x}_n)\\
&=\sum_{\mathbf{x}_1\in H(\mathbf{x}_0)}\frac{1}{\binom{|S(\mathbf{x}_0)|-1}{|\mathbf{x}_1|-1}}\frac{\mathrm{Cost}(T_{\mathbf{x}_1})}{\mathrm{Cost}(T_{S(\mathbf{x}_0)})}\max_{\substack{(\mathbf{x}_2,\dots,\mathbf{x}_n) \\ \in\mathcal{X}(\mathbf{x}_1)}}\alpha(\mathbf{x}_0,\mathbf{x}_1,\dots,\mathbf{x}_n)\\
&\leq\sum_{\mathbf{x}_1\in H(\mathbf{x}_0)}\frac{1}{\binom{|S(\mathbf{x}_0)|-1}{|\mathbf{x}_1|-1}}\frac{\mathrm{Cost}(T_{\mathbf{x}_1})}{\mathrm{Cost}(T_{S(\mathbf{x}_0)})}\max_{\mathbf{x}_1\in H(\mathbf{x}_0)}\max_{\substack{(\mathbf{x}_2,\dots,\mathbf{x}_n) \\ \in\mathcal{X}(\mathbf{x}_1)}}\alpha(\mathbf{x}_0,\mathbf{x}_1,\dots,\mathbf{x}_n)\\
&=\sum_{\mathbf{x}_1\in H(\mathbf{x}_0)}\frac{1}{\binom{|S(\mathbf{x}_0)|-1}{|\mathbf{x}_1|-1}}\frac{\mathrm{Cost}(T_{\mathbf{x}_1})}{\mathrm{Cost}(T_{S(\mathbf{x}_0)})}\max_{\substack{(\mathbf{x}_1,\dots,\mathbf{x}_n) \\ \in\mathcal{X}(\mathbf{x}_0)}}\alpha(\mathbf{x}_0,\mathbf{x}_1,\dots,\mathbf{x}_n)\\
&=\max_{\substack{(\mathbf{x}_1,\dots,\mathbf{x}_n) \\ \in\mathcal{X}(\mathbf{x}_0)}}\alpha(\mathbf{x}_0,\mathbf{x}_1,\dots,\mathbf{x}_n)\sum_{\mathbf{x}_1\in H(\mathbf{x}_0)}\frac{1}{\binom{|S(\mathbf{x}_0)|-1}{|\mathbf{x}_1|-1}}\frac{\mathrm{Cost}(T_{\mathbf{x}_1})}{\mathrm{Cost}(T_{S(\mathbf{x}_0)})}\\
\end{aligned}
\end{multline*}

Due to Lemma \ref{sumlemma}, the summation simplifies to 1, which yields the proposition.
\qed\end{proof}

In case the last bound was still too difficult to calculate because it requires taking a maximum over a very large set, we now relax the bound to a product of maximums over smaller sets.

\begin{corollary}
For a forest $T_{\mathbf{v}}$, let $\mathcal{L}_i(\mathbf{v})$ be the collection of all possible hypernodes $\mathbf{x}_i$ which can be chosen by Algorithm 2 at level $i$ of the tree, where level 0 is the root $\mathbf{v}$.  Also let $\mathbf{x}_0=\mathbf{v}$.  Then
\[\mathrm{CV}^2\big(|\mathbf{v}|C_{\mathrm{SEI}}(T_\mathbf{v})\big)\leq\prod_{i=0}^{n-1}\max_{\substack{\mathbf{x}_i\in\mathcal{L}_i \\ \mathbf{x}_{i+1}\in H(\mathbf{x}_i)}}\alpha(\mathbf{x}_{i},\mathbf{x}_{i+1})-1\]
\end{corollary}
\begin{proof}
Since
\[\alpha(\mathbf{x}_0,\mathbf{x}_1,\dots,\mathbf{x}_n)=\prod_{i=0}^{n-1}\alpha(\mathbf{x}_{i},\mathbf{x}_{i+1})\]
this follows directly from Theorem \ref{maxalpha} by expanding the set over which the maximum is taken.
\qed\end{proof}

\section{Application: Counting Linear Extensions}

\subsection{Background}
For a partially ordered set (poset), any total order which respects the partial order is known as a linear extension.  If we view the poset as a directed acyclic graph (DAG), then a linear extension is also known as a topological sort.  Determining the number of linear extensions of a given poset is a fundamental problem in the study of ordering and has many applications.

The standard method of obtaining a linear extension, which was first described by \citet{kahn}, corresponds to a breadth-first search of the corresponding DAG.  The procedure is as follows.  We choose some maximal element of the poset, and then delete that element from the poset, thus newly rendering some elements maximal.  We repeat until there are no elements left in the poset.  The order in which the elements were deleted is then a linear extension of the poset.

In fact, all linear extensions can be obtained in this manner, so the number of linear extensions is equal to the number of ways to execute the procedure.  We can think of the choices available to us at each step as forming a decision tree, where each branching corresponds to a choice of a maximal element, and each path from root to leaf corresponds to one linear extension.  If we let the cost function on the decision tree be one on the leaves and zero everywhere else, then the cost of the entire tree is the number of linear extensions.

\begin{example}
\label{LXexample}
Consider the poset whose Hasse diagram is shown in the left panel of Figure \ref{LXexamplefigure}.  The corresponding decision tree is shown in the right panel.  Each node in the decision tree is labeled according to which maximal node in the poset was deleted in order to continue the linear extension.  This poset has seven linear extensions, as shown by the decision tree.
\begin{figure}[h]
\begin{minipage}{0.5\textwidth}
\centering
\begin{tikzpicture}
\tikz \graph [no placement, nodes = {circle,draw}]{
a[x=0,y=2] -- c[x=1,y=1];
b[x=2,y=2] -- c;
b -- d[x=3,y=1];
c -- e[x=1,y=0];
};
\end{tikzpicture}
\end{minipage}%
\begin{minipage}{0.5\textwidth}
\centering
\begin{tikzpicture}[
level distance = 10mm,
level 1/.style={sibling distance=35mm},
level 2/.style={sibling distance=15mm},
level 3/.style={sibling distance=12mm},
level 4/.style={sibling distance=8mm}]
\coordinate
child {
	node {a}
	child {
		node {b}
		child {
			node {c}
			child {node {d} child {node (L1) {e}}}
			child {node {e} child {node {d}}}
			}
		child {node (L2) {d} child {node {c} child {node {e}}}}
		}
	}
child {
	node {b}
	child {
		node {a}
		child {
			node {c}
			child {node {d} child {node (R1) {e}}}
			child {node{e} child {node {d}}}
			}
		child {node (R2) {d} child {node {c} child {node {e}}}}
		}
	child {
		node {d}
		child {node {a} child {node {c} child {node {e}}}}
		}
	};
\node[draw=blue, fit=(L1)(L2), inner sep=2mm,rectangle,rounded corners] {};
\node[draw=blue, fit=(R1)(R2), inner sep=2mm,rectangle,rounded corners] {};
\end{tikzpicture}
\end{minipage}
\caption{The Hasse diagram of a poset (left) and its corresponding decision tree (right)}\label{LXexamplefigure}
\end{figure}

Notice that the two boxed regions of the decision tree are identical.  This is because they both occur after the deletion of $a$ and $b$, in either order.  There are other identical regions, as well, such as those that occur after deletion of $a,b,d$ in any of the three possible orders.\qed
\end{example}

Aside from posets which are already linearly ordered, all linear extension decision trees contain identical regions such as those highlighted in the example.  For example, in any poset with at least two maximal elements, $a$ and $b$, the region of the decision tree corresponding to extensions beginning with $ab$ will be identical to the region of the decision tree corresponding to extensions beginning with $ba$.

Hence, even if a poset is generated randomly, the linear extension decision tree arising from the poset will not be random because many regions within it are not independent of one another.  If the decision tree were completely random, there would be no reason to place more importance on any one particular branch than another.  Since the tree is not completely random, however, we have reason to suppose that we may benefit by introducing an importance function.

\subsection{Importance Functions}

Three different importance functions were tested, all of which were compared to the uniform importance function, which is 1 on all nodes.  In order to describe the three importance functions, we must first define some notation.

Let $n$ denote the number of elements in the poset.  For a node $x$ in a decision tree, let $\mathrm{sib}(x)$ denote the number of siblings of $x$ in the tree; that is, the number of nodes which are children of the parent of $x$.  Let $\mathrm{height}(x)$ denote the height of $x$ in the tree, so that leaves have height 0 and the root has height $n$.

Each node in the decision tree corresponds to some element in the poset, as shown, for example, in Figure \ref{LXexamplefigure}.  For a node $x$ in a decision tree, let $\mathrm{desc}(x)$ be the number of elements in the poset which are descendants of the poset element that $x$ corresponds to.  Note that $\mathrm{desc}(x)$ also counts whichever poset element $x$ corresponds to, so that we always have $\mathrm{desc}(x)\geq1$.

Then the three importance functions are as follows.

Importance function 1:
\[r(x)=\mathrm{sib}(x)^3\]

Importance function 2:
\[r(x)=\mathrm{sib}(x)^3\mathrm{desc}(x)\]

Importance function 3:
\[r(x)=\mathrm{sib}(x)^3\left(\frac{\mathrm{height}(x)+\mathrm{desc}(x)}{\mathrm{height}(x)-\mathrm{desc}(x)}\right)\]

It is worthwhile to note that when the budget $B=1$, each hypernode $\mathbf{x}_i$ chosen by Algorithm 2 is a single node, and so each $S(\mathbf{x}_i)$ is a single sibling set.  Hence the value of $\mathrm{sib}(x)$ is uniform for all $x\in S(\mathbf{x}_i)$ for all $\mathbf{x}_i$.  Thus $\mathrm{sib}(x)$ will cancel out of all instances of the quotient
\[\frac{r(S(\mathbf{x}_{i}))}{r(\mathbf{x}_{i+1})}\]
and have no effect on the importance function.  Therefore, when the budget $B=1$, importance function 1 is the same as the uniform importance function, as can be seen in Figure \ref{graph1}, and importance functions 2 and 3 are the same as those described in \citet*{ourLXpaper}, where their properties are explored in more detail.

\subsection{Methodology and Results}

All numerical tests of Algorithm 2 were implemented in C++ using a sparse representation of the posets.  All posets were randomly generated in the following manner.  Given the poset elements $v_1,v_2,\dots,v_n$, for each pair of elements $v_i$ and $v_j$ with $i<j$, the relation $v_i>v_j$ was given a 20\% probability to exist using a pseudo-random number generator.  The posets were then transitively completed.

The first set of tests compared the relative variance of the four importance functions as a function of the size of the posets.  These tests were repeated for the fixed budget values $B=1, 5, 10, 15, 20$.  At each budget size, the size of the poset, $n$, ran through the values $10,15,20,\dots,85$.

For each value of $B$ and $n$, $n^2$ posets were generated, and $n^2$ estimates were performed on each poset to calculate the relative variance for that poset.  The relative variance of for each poset was then averaged for each value of $n$.

The results for the different importance functions were of differing orders of magnitude; therefore, they are compared on a log-log scale.  The results are shown in Figures \ref{graph1}-\ref{graph5}.

The second set of tests compared the relative variance of the four importance functions as a function of the budget $B$.  These tests were repeated for the fixed poset sizes $n=10, 20, 40$.  At each poset size, the budget, $B$, ran through the values $1,2,3,\dots,100$.

Just as in the first set of tests, for each value of $B$ and $n$, $n^2$ posets were generated, and $n^2$ estimates were performed on each poset to calculate the relative variance for that poset.  The relative variance of each poset was then averaged for each value of $B$.

The results drop sharply as the budget grows; therefore, they are compared on a semi-log scale.  The results are shown in Figures \ref{graph6}-\ref{graph8}.

As was noted by \citet*{vaisman}, it is possible for an importance function to make the variance worse, as we see in Figure \ref{graph6} in posets with 10 nodes using importance function 1.  However, at all of the poset sizes tested, importance functions 2 and 3 present significant improvements to the variance, and this serves as evidence that it is worthwhile to incorporate an importance function into Stochastic Enumeration.

The importance function which performed the best was importance function 3, with a few exceptions.  For posets with less than 20 elements and a small budget of 1 or 2, importance function 2 performed slightly better.  In general, the improvements in variance due to all importance functions decreased slightly as the budget increased.

\begin{figure}
\centering
\includegraphics[width=\textwidth]{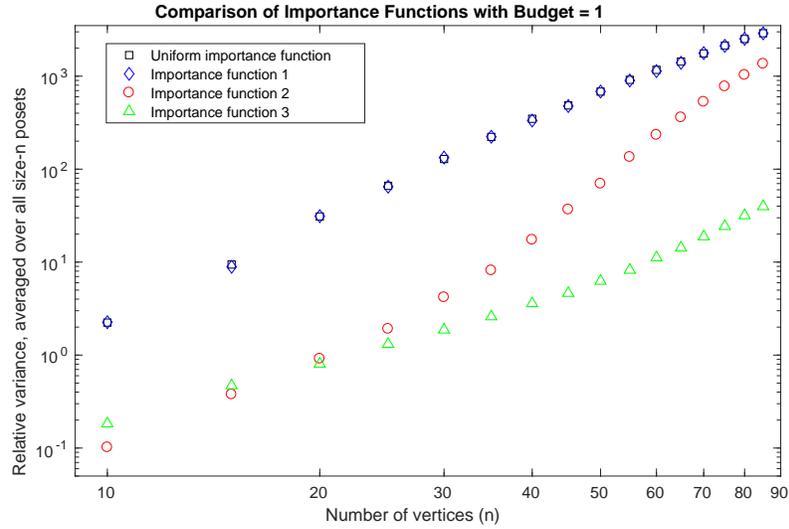}
\caption{Comparison of the relative variance of the importance functions as a function of poset size for fixed budget $B=1$}\label{graph1}
\end{figure}

\begin{figure}
\centering
\includegraphics[width=\textwidth]{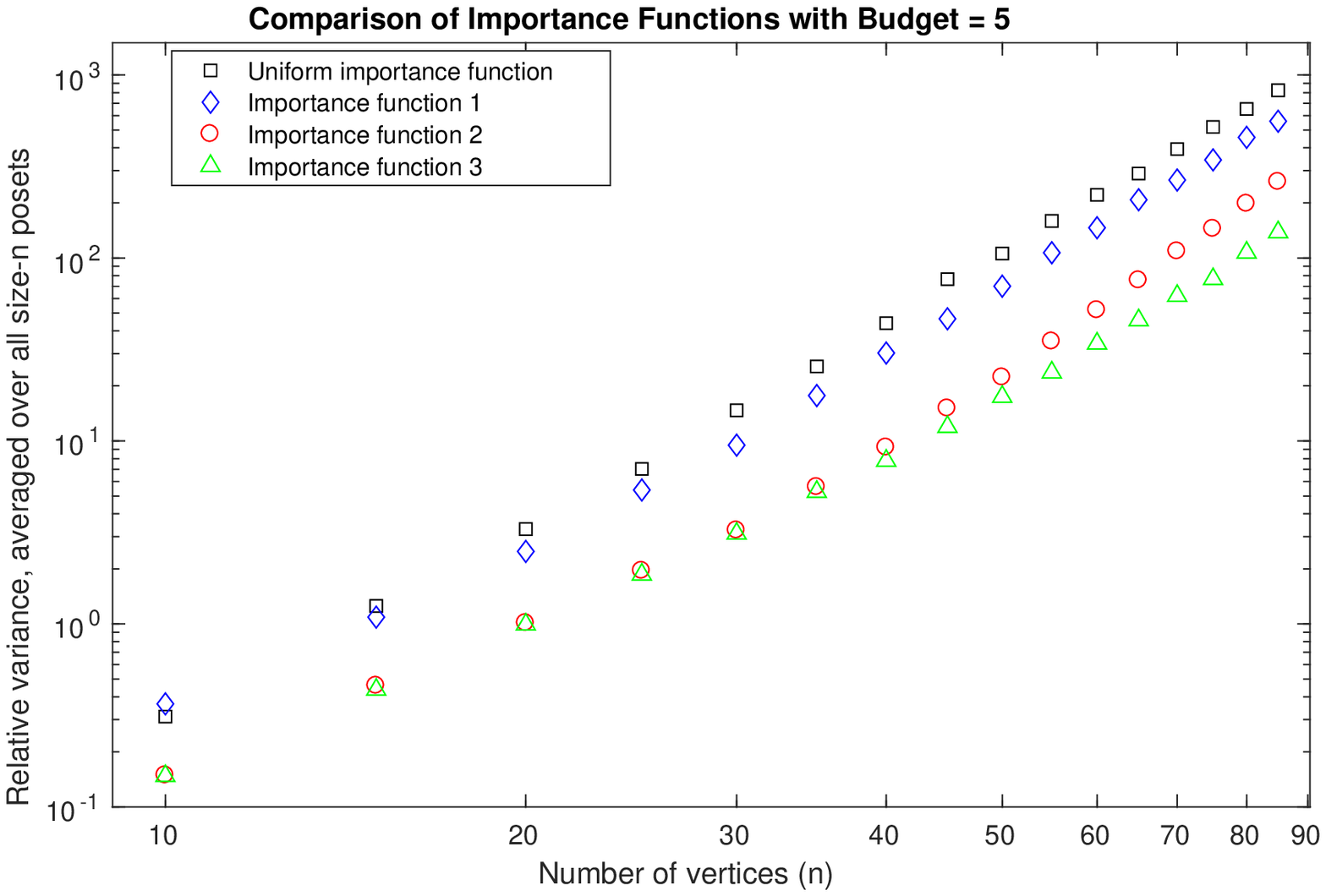}
\caption{Comparison of the relative variance of the importance functions as a function of poset size for fixed budget $B=5$}\label{graph2}
\end{figure}

\begin{figure}
\centering
\includegraphics[width=\textwidth]{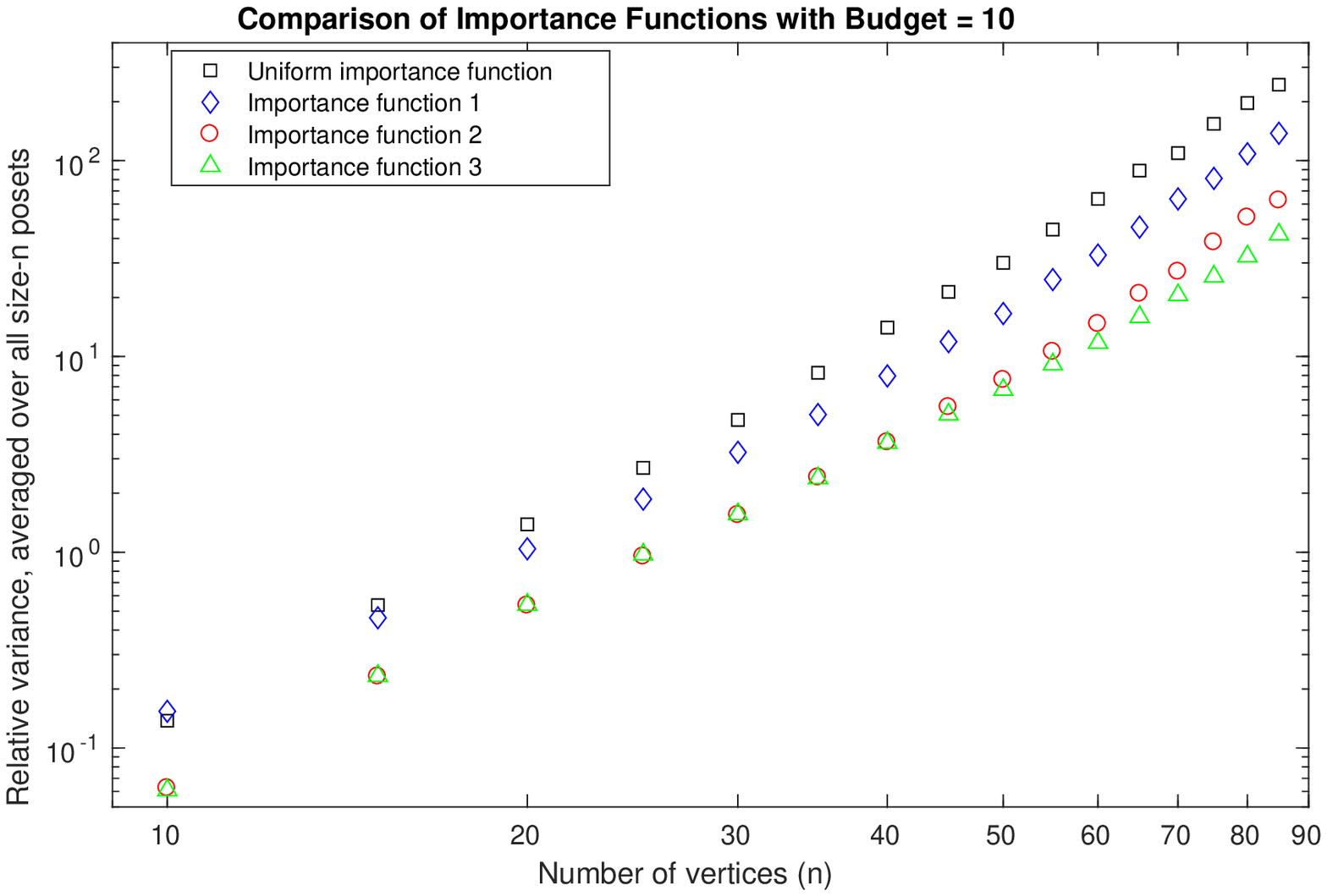}
\caption{Comparison of the relative variance of the importance functions as a function of poset size for fixed budget $B=10$}\label{graph3}
\end{figure}

\begin{figure}
\centering
\includegraphics[width=\textwidth]{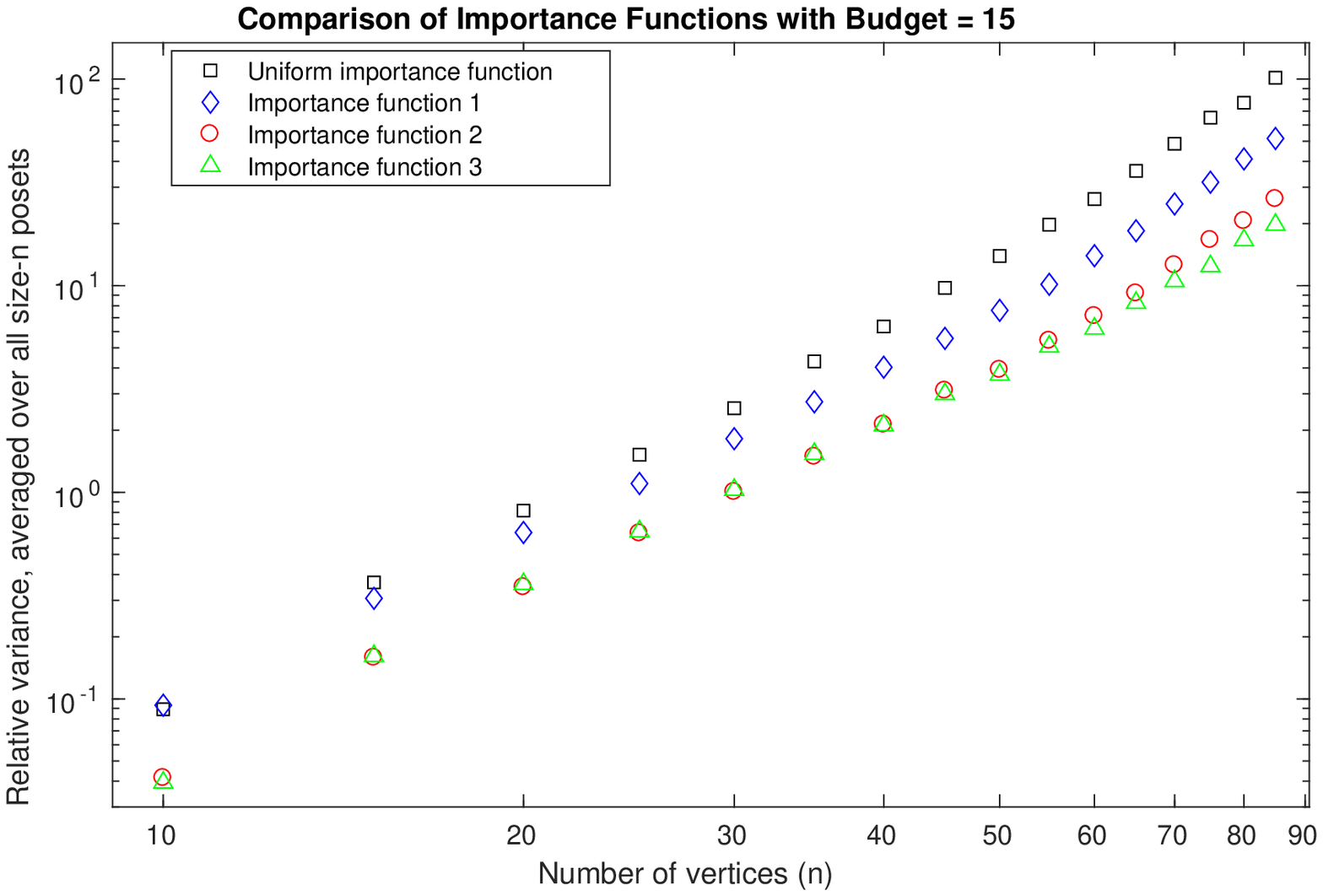}
\caption{Comparison of the relative variance of the importance functions as a function of poset size for fixed budget $B=15$}\label{graph4}
\end{figure}

\begin{figure}
\centering
\includegraphics[width=\textwidth]{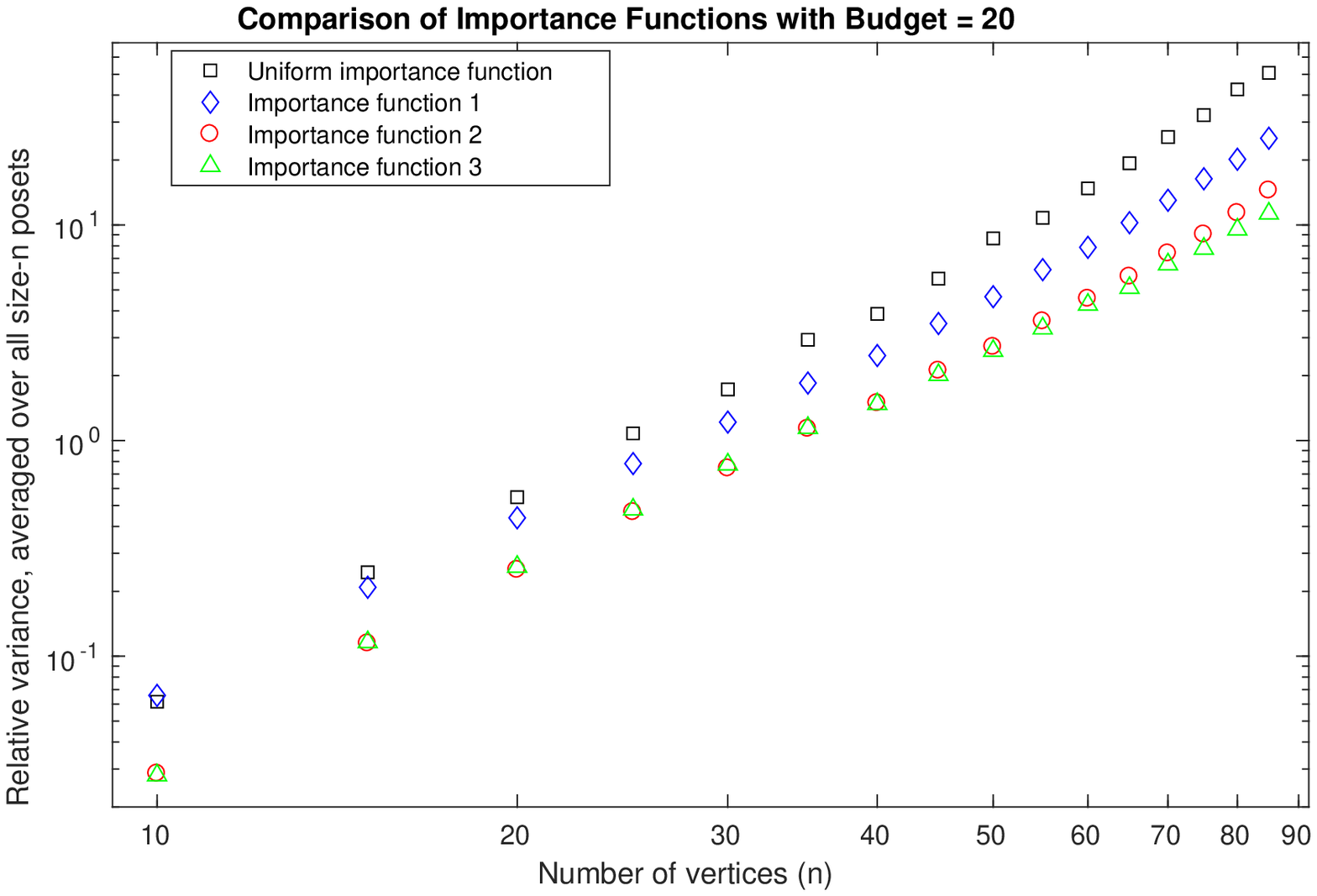}
\caption{Comparison of the relative variance of the importance functions as a function of poset size for fixed budget $B=20$}\label{graph5}
\end{figure}

\begin{figure}
\centering
\includegraphics[width=\textwidth]{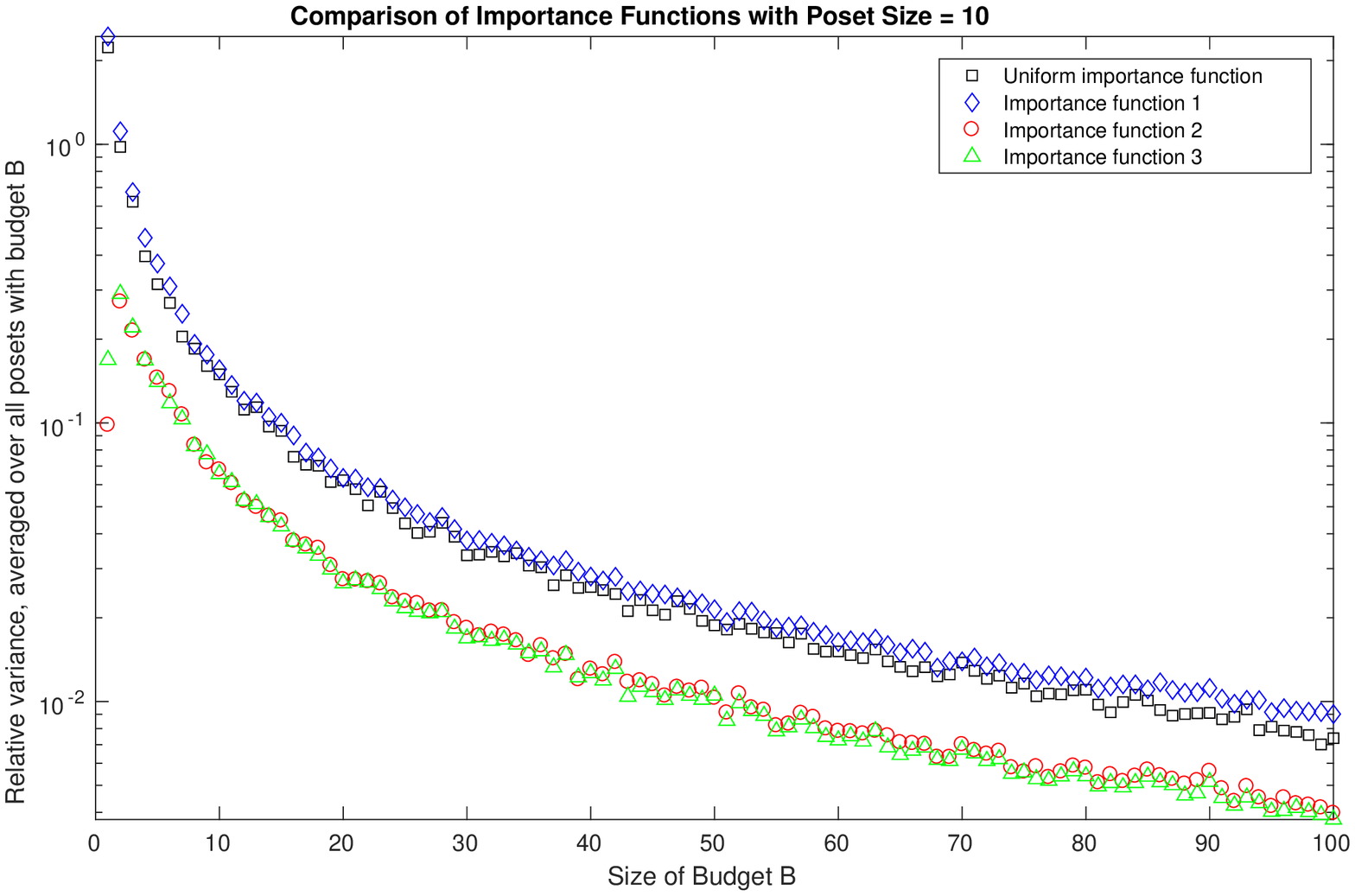}
\caption{Comparison of the relative variance of the importance functions as a function of budget size for fixed poset size $n=10$}\label{graph6}
\end{figure}

\begin{figure}
\centering
\includegraphics[width=\textwidth]{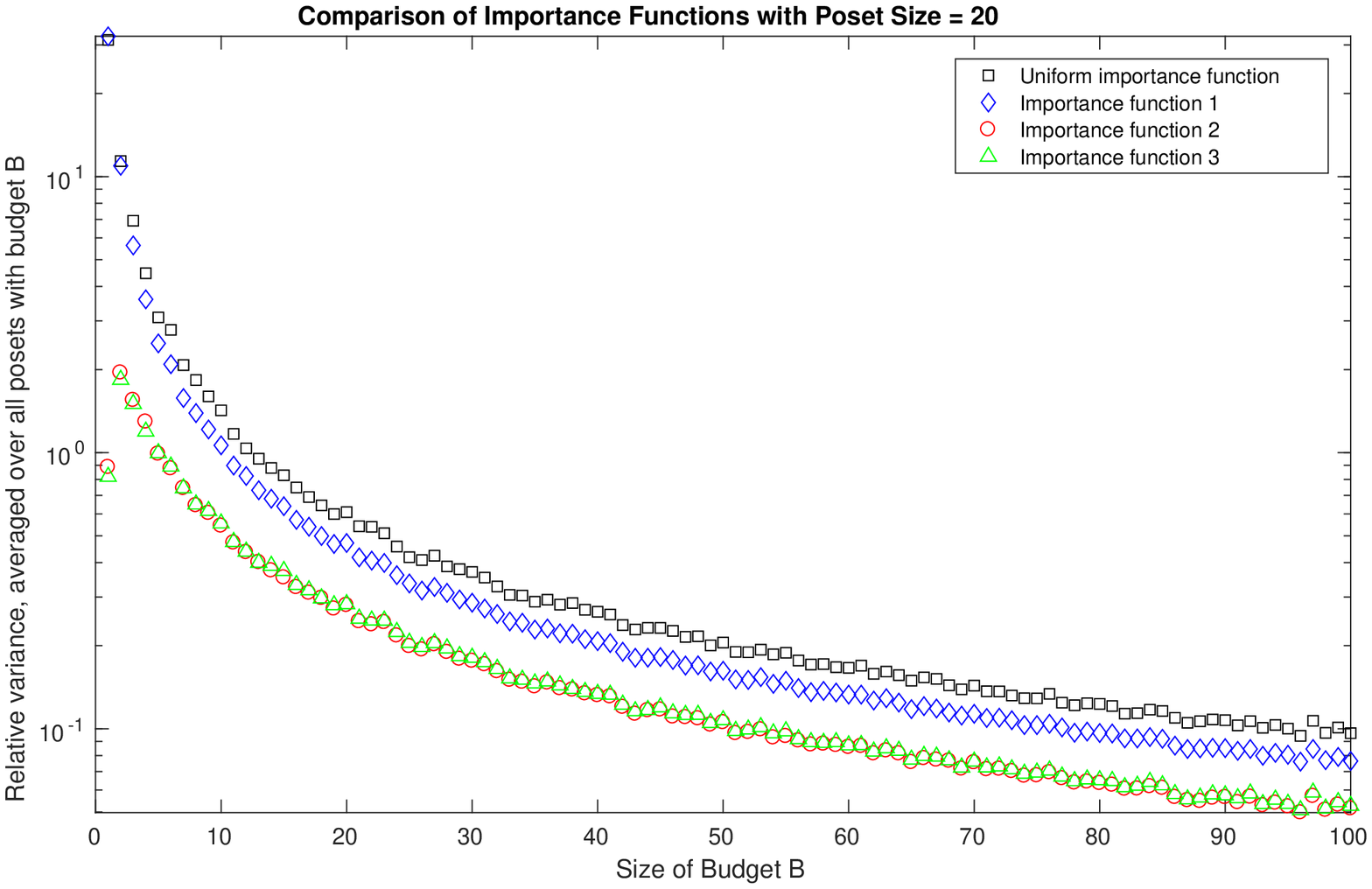}
\caption{Comparison of the relative variance of the importance functions as a function of budget size for fixed poset size $n=20$}\label{graph7}
\end{figure}

\begin{figure}
\centering
\includegraphics[width=\textwidth]{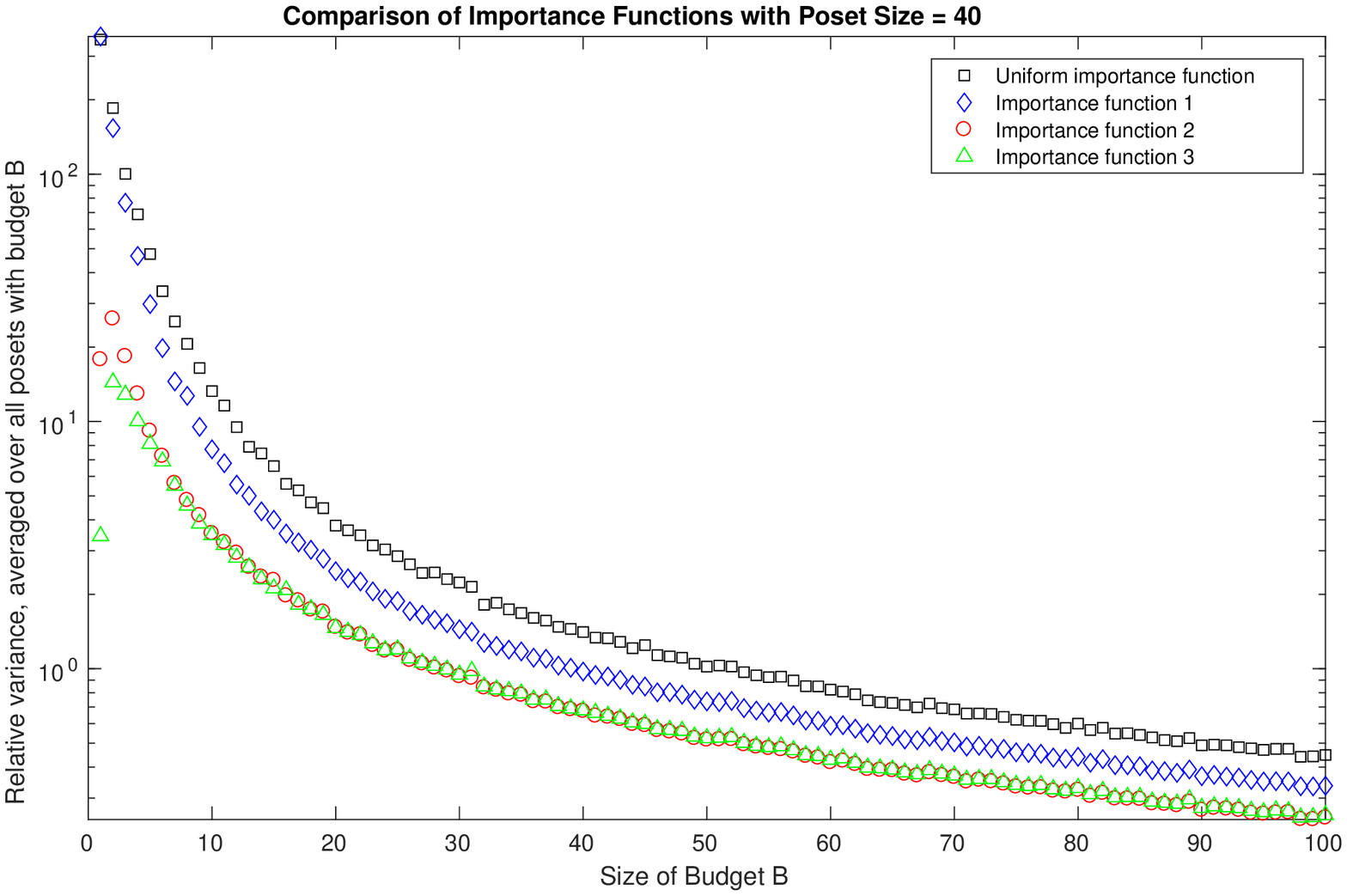}
\caption{Comparison of the relative variance of the importance functions as a function of budget size for fixed poset size $n=40$}\label{graph8}
\end{figure}

\section{Conclusions and Future Work}

We designed and implemented two generalizations of the stochastic enumeration method for counting the leaves of trees: the first algorithm for any user-supplied probability distribution on hypernodes, and the second algorithm for a probability distribution induced by any user-supplied importance function on the nodes of the tree.  We numerically tested the second algorithm on the problem of counting linear extensions of random posets, and showed that introducing an importance function can significantly reduce the variance of estimates.

Although our importance functions performed well in numerical testing, in future, we would like to explore the question of how to find better importance functions, as well as importance functions with provable performance guarantees.

\bibliographystyle{spbasic}%spmpsci}      % mathematics and physical sciences
\bibliography{seibib}   % name your BibTeX data base

\begin{thebibliography}{13}
\providecommand{\natexlab}[1]{#1}
\providecommand{\url}[1]{{#1}}
\providecommand{\urlprefix}{URL }
\expandafter\ifx\csname urlstyle\endcsname\relax
  \providecommand{\doi}[1]{DOI~\discretionary{}{}{}#1}\else
  \providecommand{\doi}{DOI~\discretionary{}{}{}\begingroup
  \urlstyle{rm}\Url}\fi
\providecommand{\eprint}[2][]{\url{#2}}

\bibitem[{Beichl and Sullivan(1999)}]{jcp}
Beichl I, Sullivan F (1999) Approximating the permanent via importance sampling
  with applications to the dimer covering problem. Journal of Computational
  Physics 149:128--147

\bibitem[{Beichl et~al(2017)Beichl, Jensen, and Sullivan}]{ourLXpaper}
Beichl I, Jensen A, Sullivan F (2017) A sequential importance sampling
  algorithm for estimating linear extensions, preprint

\bibitem[{Blitzstein and Diaconis(2011)}]{blitz}
Blitzstein J, Diaconis P (2011) A sequential importance sampling algorithm for
  generating random graphs with prescribed degrees. Internet Mathematics
  6(4):489--522

\bibitem[{Chen(1992)}]{chen}
Chen PC (1992) Heuristic sampling: A method for predicting the performance of
  tree searching programs. SIAM Journal on Computing 21(2):295--315

\bibitem[{Cloteaux and Valentin(2011)}]{brian}
Cloteaux B, Valentin LA (2011) Counting the leaves of trees. Congressus
  Numerantium 207:129--139

\bibitem[{Harris et~al(2014)Harris, Sullivan, and Beichl}]{algo}
Harris D, Sullivan F, Beichl I (2014) Fast sequential importance sampling to
  estimate the graph reliability polynomial. Algorithmica 68(4):916--939

\bibitem[{Kahn(1962)}]{kahn}
Kahn AB (1962) Topological sorting of large networks. Commun ACM
  5(11):558--562, \doi{10.1145/368996.369025}

\bibitem[{Karp and Luby(1983)}]{karp}
Karp RM, Luby M (1983) Monte-carlo algorithms for enumeration and reliability
  problems. In: Proceedings of the 24th Annual Symposium on Foundations of
  Computer Science, SFCS '83, IEEE Computer Society, pp 56--64

\bibitem[{Knuth(1975)}]{knuth}
Knuth DE (1975) Estimating the efficiency of backtrack programs. Mathematics of
  Computation 29(129):121--136

\bibitem[{Rubinstein(2013)}]{rubin}
Rubinstein R (2013) Stochastic enumeration method for counting {NP}-hard
  problems. Methodology and Computing in Applied Probability 15(2):249--291

\bibitem[{Rubinstein et~al(2014)Rubinstein, Ridder, and Vaisman}]{rubinbook}
Rubinstein R, Ridder A, Vaisman R (2014) Fast sequential Monte Carlo methods
  for counting and optimization. Wiley

\bibitem[{Vaisman and Kroese(2017)}]{vaisman}
Vaisman R, Kroese DP (2017) Stochastic enumeration method for counting trees.
  Methodology and Computing in Applied Probability 19(1):31--73,
  \doi{10.1007/s11009-015-9457-4}

\bibitem[{Valiant(1979)}]{valiant}
Valiant LG (1979) The complexity of enumeration and reliability problems. SIAM
  Journal on Computing 8(3):410--421

\end{thebibliography}

\end{document}